\documentclass[11pt]{article}

\usepackage{graphicx,amssymb,amsmath}
\usepackage{amsfonts}
\usepackage{booktabs}
\usepackage{color}
\usepackage{enumerate}
\usepackage{float}

\usepackage{cite}
\usepackage[pdftex, plainpages = false, pdfpagelabels, 
                 bookmarks=false,
                 bookmarksopen = true,
                 bookmarksnumbered = true,
                 breaklinks = true,
                 linktocpage,
                 pagebackref,
                 colorlinks = true,  
                 linkcolor = blue,
                 urlcolor  = cyan,
                 citecolor = red,
                 anchorcolor = green,
                 hyperindex = true,
                 hyperfigures
                 ]{hyperref}

\graphicspath{{./figures/}}
\usepackage{mathrsfs}
\usepackage{subfig}
\usepackage{tabularx}
\usepackage{authblk}

\usepackage[in]{fullpage}
\usepackage[top=0.8in, bottom=0.9in, left=0.9in, right=0.9in]{geometry}

\newcommand{\poly}{\text{poly}}
\def\bbR{\mathbb{R}}
\def\calA{\mathcal{A}}
\def\calH{\mathcal{H}}
\def\calL{\mathcal{L}}
\def\calU{\mathcal{U}}
\def\calD{\mathcal{D}}
\def\calC{\mathcal{C}}
\def\hatL{\widehat{L}}
\newtheorem{observation}{Observation}
\newtheorem{definition}{Definition}
\newtheorem{corollary}{Corollary}

\newtheorem{lemma}{Lemma}
\newtheorem{theorem}{Theorem}
\newenvironment{proof}{\noindent {\textbf{Proof:}}\rm}{\hfill $\Box$ \rm\bigskip}

\title{An Optimal Algorithm for Computing Many Faces in Line Arrangements\thanks{A preliminary version of this paper will appear in {\em Proceedings of the 42nd International Symposium on Computational Geometry (SoCG 2026)}.}
}

\author{
Haitao Wang\thanks{Kahlert School of Computing,
University of Utah, Salt Lake City, UT 84112, USA. {\tt haitao.wang@utah.edu}}
}

\date{}

\begin{document}

\maketitle

\vspace{-0.35in}
\begin{abstract}
Given a set of $m$ points and a set of $n$ lines in the plane, we consider the problem of computing the faces of the arrangement of the lines that contain at least one point. In this paper, we present an $O(m^{2/3}n^{2/3}+(n+m)\log n)$ time algorithm for the problem. We also show that this matches the lower bound under the algebraic decision tree model and thus our algorithm is optimal. In particular, when $m=n$, the runtime is $O(n^{4/3})$, which matches the worst case combinatorial complexity $\Omega(n^{4/3})$ of all output faces. This is the first optimal algorithm since the problem was first studied more than three decades ago [Edelsbrunner, Guibas, and Sharir, SoCG 1988]. 
\end{abstract}

{\em Keywords:} Many faces, line arrangements, cuttings, $\Gamma$-algorithms, decision tree complexities

\section{Introduction}
\label{sec:intro}
Let $P$ be a set of $m$ points and $L$ a set of $n$ lines in the plane. We consider the problem of computing the faces of the arrangement of the lines of $L$ that contain at least one point of $P$; these faces are called {\em non-empty faces}. Note that each non-empty face only needs to be output once if it contains multiple points of $P$. Let $\calA(L)$ denote the arrangement of $L$.

\paragraph{Previous work.}
This is one of the most fundamental problems in computational geometry and has been studied extensively.  A straightforward algorithm can implicitly determine the non-empty faces of $\calA(L)$ in $O(m\log n+n^2)$ time, by first constructing $\calA(L)$, building a point location data structure for $\calA(L)$~\cite{ref:EdelsbrunnerOp86,ref:KirkpatrickOp83,ref:SarnakPl86}, and then finding the non-empty faces using the point location queries for points of $P$.  
The following upper bounds have been proved on the combinatorial
complexity of all non-empty faces of $\calA(L)$: 
$O(m^{2/3}n^{2/3}+n)$~\cite{ref:ClarksonCo90}, $O(n\sqrt{m})$~\cite{ref:EdelsbrunnerOn86}, and
$O(n+m\sqrt{n})$~\cite{ref:EdelsbrunnerOn86}. 
A lower bound of 
$\Omega(m^{2/3}n^{2/3}+n)$~\cite{ref:EdelsbrunnerOn86} is also known, which matches the $O(m^{2/3}n^{2/3}+n)$ upper bound in \cite{ref:ClarksonCo90}. Note that the minimum of the above upper bounds is at most $O(m\log n+n^2)$ regardless of the values of $m$ and $n$. Hence, using the above straightforward approach, $O(m\log n+n^2)$ time is also sufficient to output all non-empty faces explicitly. 
To have a more efficient algorithm, 
Edelsbrunner, Guibas, and Sharir~\cite{ref:EdelsbrunnerTh90} first studied the problem and gave a randomized algorithm of
$O(m^{2/3-\delta}n^{2/3+2\delta}\log n+n\log n\log m)$ expected time,
for any $\delta>0$. Subsequently, an improved deterministic algorithm of $O(m^{2/3}n^{2/3}\log^{5/3}n \log^{1.11}(m/\sqrt{n})+(m+n)\log n)$ time was proposed by Agarwal~\cite{ref:AgarwalPa902}. Also, Agarwal, Matou\v{s}ek,
and Schwarzkopf~\cite{ref:AgarwalCo98} presented a randomized
algorithm of $O(m^{2/3}n^{2/3}\log (n/\sqrt{m})+(m+n)\log n)$ expected time. No progress had been made for a while until  Wang~\cite{ref:WangCo23} revisited the problem in SODA 2022 and derived a new deterministic algorithm of $O(m^{2/3}n^{2/3}\log^{2/3} (n/\sqrt{m})+(m+n)\log n)$ time.

On the other hand, it is not difficult to see that $\Omega(m^{2/3}n^{2/3}+n\log n + m)$ is a lower bound for solving the problem  due to the above $\Omega(m^{2/3}n^{2/3}+n)$ lower bound~\cite{ref:EdelsbrunnerOn86} on the combinatorial complexity of all non-empty faces and also because computing a single face in line arrangements requires $\Omega(n\log n)$ time in the algebraic decision tree model (indeed, a special case of the problem is to compute the lower envelope of all lines, which has an $\Omega(n\log n)$ time lower bound). In particular, in the {\em symmetric case} where $m=n$, $\Omega(n^{4/3})$ is a lower bound since the worst case combinatorial complexity of all non-empty faces is $\Omega(n^{4/3})$~\cite{ref:EdelsbrunnerOn86}.

\paragraph{Our result.}
We propose a new (deterministic) algorithm of $O(m^{2/3}n^{2/3}+(m+n)\log n)$ time for the problem. In the symmetric case, the runtime of the algorithm is $O(n^{4/3})$, which matches the above $\Omega(n^{4/3})$ lower bound, and therefore our algorithm is optimal. For the {\em asymmetric case} where $m\neq n$, using Ben-Or's techniques~\cite{ref:Ben-OrLo83}, we prove that $\Omega(m\log n)$ is also a lower bound for solving the problem under the algebraic decision tree model. Combining with the above $\Omega(m^{2/3}n^{2/3}+n\log n + m)$ lower bound, we thus obtain the $\Omega(m^{2/3}n^{2/3}+(m+n)\log n)$ lower bound for solving the problem. As such, our algorithm for the asymmetric case is also optimal. 
Although our solution only improves the previously best result of Wang~\cite{ref:WangCo23} by a factor of $\log^{2/3} (n/\sqrt{m})$, the result is important for providing the first optimal solution to a fundamental and classical problem in  computational geometry. 

Our algorithm takes a different approach than the previous work. For example, the most recent work of Wang~\cite{ref:WangCo23} first solves the problem in $O(n\log n+m\sqrt{n\log n})$ time in the dual setting. Then, the algorithm is plugged into the framework of Agarwal~\cite{ref:AgarwalPa902} as a subroutine. More specifically, the framework first 
utilizes a cutting of $L$ to divide the problem into a collection of subproblems and then solves these subproblems using Wang's new algorithm.

We dispense with the framework of Agarwal~\cite{ref:AgarwalPa902}. Instead, we propose a new algorithm for the problem in the primal setting, which is different from the algorithm of Wang~\cite{ref:WangCo23} in the dual setting. One key subproblem is to merge convex hulls, for which our techniques rely on a crucial combinatorial observation that certain pairs of convex hull boundaries can intersect at most $O(1)$ times.
Using a hierarchical cutting of $L$~\cite{ref:ChazelleCu93} and combining our new algorithm with Wang's algorithm, we obtain a recursive algorithm that runs in $n^{4/3}\cdot 2^{O(\log^* n)}$ time for the symmetric case $m=n$. 

To further reduce the $2^{O(\log^*n)}$ factor, we resort to the recent techniques from Chan and Zheng~\cite{ref:ChanHo23}: the {\em $\Gamma$-algorithm framework} for bounding algebraic decision tree complexities. More specifically, after $O(1)$ recursive steps in our recursive algorithm, we reduce the problem to $O((n/b)^{4/3})$ subproblems of size $O(b)$ each (i.e., in each subproblem, we need to compute the non-empty faces of an arrangement of $b$ lines for $b$ points), with $b=O(\log\log n)$. As $b$ is very small, solving the problem efficiently under the algebraic decision tree model (i.e., only count the number of comparisons) can lead to an efficient algorithm under the conventional computational model (e.g., the real RAM model). With Chan and Zheng's techniques, we developed an algorithm that can solve each subproblem in $O(b^{4/3})$ decision tree complexity (i.e., the algorithm uses $O(b^{4/3})$ comparisons). Since $b$ is small, we are able to build a decision tree for the algorithm in $O(n)$ time. With the decision tree in hand, each subproblem can then be solved in $O(b^{4/3})$ time in the conventional computation model. Consequently, the original problem can be solved in $O(n^{4/3})$ time. 
In addition, using this algorithm as a subroutine, the asymmetric case of the problem where $m\neq n$ can be solved in $O(m^{2/3}n^{2/3}+(m+n)\log n)$ time.

It should be noted that our algorithm is not a direct application of Chan and Zheng's techniques. Nor it is a simple adaption of their algorithm for Hopcroft's problem. Indeed, in order to fit the $\Gamma$-algorithm framework, we need to design new procedures for a number of subproblems (notably, our techniques involve shaving log factors to merge convex hulls), 
which may be interesting in their own right. 
This is also the case for solving other problems when Chan and Zheng's techniques are applied. For instance,  
Chan, Cheng, and Zheng~\cite{ref:ChanAn24} recently used the technique to tackle the higher-order Voronoi diagram problem and presented an optimal algorithm by improving the previous best algorithm by a factor of $2^{O(\log^* n)}$. To this end, they had to develop new techniques and made a significant effort. On the other hand, there are other problems whose current best solutions have $2^{O(\log^* n)}$ factors, but it is not clear whether the $\Gamma$-algorithm techniques~\cite{ref:ChanHo23,ref:ChanAn24} can be used to further reduce these factors. For instance, the biclique partition problem for a set of $n$ points in the plane can be solved in $n^{4/3}\cdot 2^{O(\log^* n)}$ time~\cite{ref:KatzAn97,ref:WangIm25}. It has been open whether it is possible to have an $O(n^{4/3})$ time algorithm.

\paragraph{Related work.}
Other related problems have also been studied in the literature, e.g., the segment case where $L$ consists of $n$ line segments. Although faces in an arrangement of lines are convex, they may not even be simply connected in an arrangement of segments. Therefore, the segment case becomes more challenging. 
It has been proved that the combinatorial complexity of all non-empty faces of the segment arrangement is bounded by $O(m^{2/3}n^{2/3}+n\alpha(n)+n\log m)$~\cite{ref:AronovTh92} and $O(n\sqrt{m}\alpha(n))$~\cite{ref:EdelsbrunnerAr92}, where $\alpha(n)$ is the inverse Ackermann function; a lower bound $\Omega(m^{2/3}n^{2/3}+n\alpha(n))$~\cite{ref:EdelsbrunnerTh90} was also known. To compute all non-empty faces, 
Edelsbrunner, Guibas, and Sharir~\cite{ref:EdelsbrunnerTh90} first gave a randomized algorithm of $O(m^{2/3-\delta}n^{2/3+2\delta}\log n+n\alpha(n)\log^2 n\log m)$ expected time for any $\delta>0$. Agarwal~\cite{ref:AgarwalPa902} presented an improved deterministic algorithm of $O(m^{2/3}n^{2/3}\log n \log^{2.11}(n/\sqrt{m})+n\log^3 n+ m\log n)$ time. Also, Agarwal, Matou\v{s}ek, and Schwarzkopf~\cite{ref:AgarwalCo98} derived a randomized algorithm of
$O(n^{2/3}m^{2/3}\log^2 n+(n\alpha(n)+n\log m+m)\log n)$ expected time. 
Wang~\cite{ref:WangCo23} proposed an $O(n^{2/3}m^{2/3}\log n+\tau(n\alpha^2(n)+n\log m+m)\log n)$ time deterministic time, where $\tau=\min\{\log m,\log (n/\sqrt{m})\}$. 
An intriguing question is whether our new techniques for the line case can somehow be utilized to tackle the segment case. 
One obstacle, for example, is that the problem for the segment case in the dual setting is not ``cleanly'' defined and thus we do not have a corresponding algorithm in the dual setting. 

In part because of the difficulty of the segment case, a substantial amount of attention has been paid to a special case in which we wish to compute a single face in an arrangement of segments. The problem can be solved in 
$O(n\alpha(n)\log n)$ expected time using a randomized algorithm~\cite{ref:ChazelleCo93}, or in $O(n\alpha^2(n)\log n)$ time using a deterministic algorithm~\cite{ref:AmatoCo95}. These algorithms provide improvements over the previous deterministic methods, which required $O(n\log^2 n)$ time~\cite{ref:MitchellOn90} and $O(n\alpha(n)\log^2 n)$ time~\cite{ref:EdelsbrunnerTh90} respectively.
An open problem in this field has been whether it is possible to have an $O(n\alpha(n)\log n)$ time deterministic algorithm. It is worth noting that computing the upper envelope of all segments can be accomplished more efficiently in $O(n\log n)$ time~\cite{ref:HershbergerFi89}.

In addition, the online query problems have also been studied, i.e., preprocess a set of lines or segments so that the face containing a query point can be quickly reported; see, e.g., \cite{ref:EdelsbrunnerIm89,ref:GuibasCo91,ref:WangCo23}.

\paragraph{Outline.}
The rest of the paper is organized as follows. After introducing some notation and concepts in Section~\ref{sec:pre},
we present an algorithm in Section~\ref{sec:primal} that solves the problem in the primal plane. 
Then, in Section~\ref{sec:dual}, we describe the second algorithm that tackles the problem in the dual plane; this algorithm is essentially Wang's algorithm in \cite{ref:WangCo23} but we slightly modify it to make it recursive. Note that our first algorithm in Section~\ref{sec:primal} is new. We combine the two algorithms in Section~\ref{sec:combine} to obtain our final algorithm; to this end, new procedures under the $\Gamma$-algorithm framework~\cite{ref:ChanHo23} are also developed. The lower bound is proved at the very end of  Section~\ref{sec:combine}.

\section{Preliminaries}
\label{sec:pre}

We follow the same notation as in Section~\ref{sec:intro}, e.g., $L$, $P$, $n$, $m$, $\calA(L)$. 
Our goal is to compute all nonempty cells of $\calA(L)$. 

For ease of discussion, we assume that no point of $P$ lies on a line of $L$. Note that this implies that 
each point of $P$ is in the interior of a face of $\calA(L)$. 
We also make a general position assumption that no line of $L$ is vertical. 
These assumptions can be relaxed without difficulty by the standard perturbation techniques~\cite{ref:EdelsbrunnerSi90,ref:YapSy90}. 

For any point $p$, denote by $F_p(L)$ the face of the arrangement $\calA(L)$ that contains $p$.

For a compact region $R$ in the plane, we often use $P(R)$ to denote the subset of $P$ in $R$, i.e., $P(R)=P\cap R$.

\paragraph{Cuttings.}
A tool that will be frequently used in our algorithm (as well as in the previous work) is cuttings~\cite{ref:ChazelleCu93,ref:MatousekRa93}.
For a compact region $R$ in the plane, we use $L_R$ to denote the subset of lines of $L$ that intersect the interior of $R$ (we also say that these lines {\em cross} $R$ and $L_R$ is the {\em conflict list} of $R$).

A {\em cutting} for $L$ is a collection $\Xi$ of closed cells (each of which is a possibly unbounded triangle) with disjoint interiors, which together cover the entire plane~\cite{ref:ChazelleCu93,ref:MatousekRa93}.
The {\em size} of $\Xi$ is defined to be the number of cells in $\Xi$.
For a parameter $r$ with $1\leq r\leq n$, a cutting $\Xi$ for $L$ is a {\em $(1/r)$-cutting} if $|L_{\sigma}|\leq n/r$ holds for every cell $\sigma\in \Xi$.


We say that a cutting $\Xi'$ {\em $c$-refines} another cutting $\Xi$ if each cell of $\Xi'$ is wholly contained within a single cell of $\Xi$, and if every cell in $\Xi$ encompasses at most $c$ cells from $\Xi'$. 

A hierarchical {\em $(1/r)$-cutting}, characterized by constants $c$ and $\rho$, consists of a series of cuttings $\Xi_0, \Xi_1, \ldots, \Xi_k$ with the following property. $\Xi_0$ has a single cell that is the entire plane. For each $1 \leq i \leq k$, $\Xi_i$ is a $(1/\rho^i)$-cutting of size $O(\rho^{2i})$ that $c$-refines $\Xi_{i-1}$. To make $\Xi_k$ a $(1/r)$-cutting, we select $k = \Theta(\log r)$ to ensure $\rho^{k-1} < r \leq \rho^k$. Consequently, the size of $\Xi_k$ is $O(r^2)$. As already shown in \cite{ref:ChazelleCu93}, it can be easily verified that the total number of cells of all cuttings $\Xi_i$, $0\leq i\leq k$, is also $O(r^2)$, and the total size of the conflict lists $L_{\sigma}$ for all cells $\sigma$ of $\Xi_i$, $0\leq i\leq k$, is bounded by $O(nr)$. 
If a cell $\sigma$ within $\Xi_{i-1}$ contains a cell $\sigma'$ in $\Xi_i$, we call $\sigma$ the {\em parent} of $\sigma'$ and $\sigma'$ a {\em child} of $\sigma$. 
In the following, we often use $\Xi$ to denote the set of all cells of all cuttings $\Xi_i$, $0\leq i\leq k$.

For any $1\leq r\leq n$, a hierarchical $(1/r)$-cutting of size $O(r^2)$ for $L$ (together with the conflict lists $L_{\sigma}$ for all cells $\sigma$ of $\Xi_i$ for all $i=0,1,\ldots,k$) can be computed in
$O(nr)$ time by Chazelle's algorithm~\cite{ref:ChazelleCu93}.




\section{The first algorithm in the primal plane}
\label{sec:primal}
In this section, we present our first algorithm for the problem, which works in the primal plane (this is in contrast to the second algorithm described in Section~\ref{sec:dual}, which tackles the problem in the dual plane). 

For any subset $L'\subseteq L$, let $\calU(L')$ denote the upper envelope of the lines of $L'$. 

Consider a point $p\in P$. Our algorithm needs to output the face $F_p(L)$ of $\calA(L)$ that contains $p$. Let $L_+(p)$ (resp., $L_-(p)$) denote the subset of lines of $L$ that are below (resp., above) $p$. 
It is not difficult to see that the face $F_p(L)$ is the common intersection of the region above the upper envelope $\calU(L_+(p))$ and the region below the lower envelope of $L_-(p)$; see Fig.~\ref{fig:primal}. Our algorithm will compute binary search trees of height $O(\log n)$ that represent $\calU(L_+(p))$ and the lower envelope of $L_-(p)$, respectively. Using the two trees, $F_p(L)$ can be computed in $O(\log n)$ time by computing the two intersections between $\calU(L_+(p))$ and $\calU(L_-(p))$, and $F_p(L)$ can then be output in additional $O(|F_p(L)|)$ time. In what follows, we focus on computing $\calU(L_+(p))$ since the lower envelope of $L_-(p)$ can be treated likewise. In the following discussion, depending on the context, an upper envelope (or lower envelope) may refer to a binary search tree that represents it. For example, the phrase ``computing $\calU(L_+(p))$'' means ``computing a binary search tree that represents $\calU(L_+(p))$''.

\begin{figure}[t]
\begin{minipage}[t]{\textwidth}
\begin{center}
\includegraphics[height=1.4in]{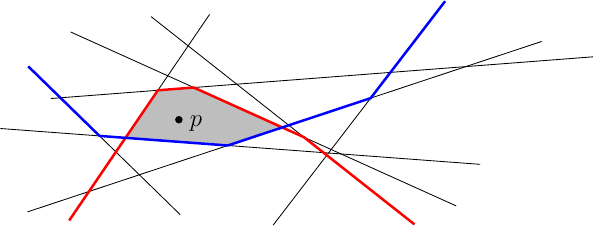}
\caption{\footnotesize Illustrating $F_p(L)$ (the grey cell), $\calU(L_+(p))$ (blue segments), and the lower envelope of $L_-(p)$ (red segments).}
\label{fig:primal}
\end{center}
\end{minipage}
\vspace{-0.15in}
\end{figure}

We start with computing a hierarchical $(1/r)$-cutting $\Xi_0,\Xi_1,\ldots,\Xi_k$ for the lines of $L$ in $O(nr)$ time~\cite{ref:ChazelleCu93}, for a parameter $r\in [1,n]$ to be determined later. Let $\Xi$ denote the set of all cells $\sigma\in \Xi_i$, $0\leq i\leq k$. The algorithm also produces the conflict lists $L_{\sigma}$ for all cells $\sigma\in \Xi$. 
For each cell $\sigma\in \Xi_i$, $1\leq i\leq k$, we define $\hatL_{\sigma}$ as the subset of lines of $L_{\sigma'}$ that are completely below $\sigma$, where $\sigma'$ is the parent of $\sigma$ (and thus $\sigma'$ is a cell of $\Xi_{i-1}$). 
As $L_{\sigma'}$ is already available and $\sigma'$ has $O(1)$ children, computing $\hatL_{\sigma}$ for all children $\sigma$ of $\sigma'$ can be done in $O(|L_{\sigma'}|)$ time by brute force. As such, computing $\hatL_{\sigma}$ for all cells $\sigma\in \Xi$ takes $O(nr)$ time. 

We wish to maintain a binary search tree to represent the upper envelope $\calU(\hatL_{\sigma})$ of $\hatL_{\sigma}$. To compute it in linear time, we need to have lines of $\hatL_{\sigma}$ in sorted order by their slopes. If we sort each $\hatL_{\sigma}$ individually, the total sorting time for all cells $\sigma\in \Xi$ is $O(nr\log n)$. To reduce the $\log n$ factor, we adopt the following strategy. We presort all lines of $L$ in $O(n\log n)$ time. Then, we have the following lemma. 

\begin{lemma}
Suppose we have a sorted list of the lines of $L$ by their slopes. Then, $\hatL_{\sigma}$ for all cells $\sigma\in \Xi$ can be sorted in $O(nr)$ time. 
\end{lemma}
\begin{proof}
For each line $\ell\in L$, we maintain a linked list $A_{\ell}$ of all cells $\sigma\in \Xi$ such that $\ell\in \hatL_{\sigma}$. Computing all lists $A_{\ell}$ for all lines $\ell\in L$ can be done in $O(nr)$ time by visiting $\hatL_{\sigma}$ for each cell $\sigma\in \Xi$. Next, we ``empty'' each $\hatL_{\sigma}$ and recompute a sorted $\hatL_{\sigma}$ as follows. We process the lists $A_{\ell}$ of all lines $\ell\in L$ following their slope order. For each $A_{\ell}$, for each cell $\sigma\in A_{\ell}$, we add $\ell$ to the rear of the current $\hatL_{\sigma}$. After the lists $A_{\ell}$ of all lines $\ell\in L$ are processed as above, $\hatL_{\sigma}$ of all cells $\sigma\in \Xi$ are sorted. Clearly, the total time is proportional to the total size of $\hatL_{\sigma}$ of all cells $\sigma\in \Xi$, which is $O(nr)$. 
\end{proof}

After each $\hatL_{\sigma}$ is sorted, we compute its upper envelope $\calU(\hatL_{\sigma})$, which takes $O(|\hatL_{\sigma}|)$ time. As such, computing $\calU(\hatL_{\sigma})$ for all cells $\sigma\in \Xi$ can be done in $O(nr)$ time. 

Next, we compute $P(\sigma)$, which is defined to be $P\cap \sigma$, for all cells $\sigma\in \Xi$. This can be done in $O(m\log r)$ time by point locations in a top-down manner in the hierarchical cutting. Specifically, for each point $p\in P$, starting from $\Xi_0$, which comprises a single cell that is the entire plane, suppose the cell $\sigma$ of $\Xi_{i}$ containing $p$ is known; then since $\sigma$ has $O(1)$ cells, locating the cell of $\Xi_{i+1}$ containing $p$ can be done in $O(1)$ time. As such, performing point locations for $p$ takes $O(\log r)$ time.
Note that $\sum_{\sigma\in \Xi}|P_{\sigma}|=O(m\log r)$ as every point of $P$ is stored in a single cell of $\Xi_i$ for each $0\leq i\leq k$.

For each cell $\sigma\in \Xi_i$, $0\leq i\leq k+1$, define $L_+(\sigma)$ as the subset of the lines of $L$ completely below $\sigma$. 
We have the following lemma. 
\begin{lemma}\label{lem:20}
For each cell $\sigma\in \Xi$, suppose $\sigma'$ is the parent of $\sigma$. Then, $L_+(\sigma)=L_+(\sigma')\cup \hatL_{\sigma}$ and thus $\calU(L_+(\sigma))$ is the upper envelope of $\calU(L_+(\sigma'))$ and $\calU(\hatL_{\sigma})$. 
\end{lemma}
\begin{proof}
Let $\ell$ be a line of $L_+(\sigma)$. By definition, $\ell$ is below $\sigma$. If $\ell$ crosses $\sigma'$, then $\ell$ must be in $\hatL_{\sigma}$ by definition; otherwise, $\ell$ must be below $\sigma'$ since $\sigma\subseteq\sigma'$, and thus $\ell$ is in $L_+(\sigma')$. This proves that $L_+(\sigma)\subseteq L_+(\sigma')\cup \hatL_{\sigma}$. 

On the other hand, consider a line $\ell\in L_+(\sigma')\cup \hatL_{\sigma}$. If $\ell\in L_+(\sigma')$, then by definition, $\ell$ is below $\sigma'$. Since $\sigma\subseteq\sigma'$, $\ell$ must be below $\sigma$ and thus $\ell\in L_+(\sigma)$. If $\ell \in \hatL_{\sigma}$, by definition, $\ell$ is below $\sigma$ and thus $\ell\in L_+(\sigma)$. This proves $L_+(\sigma')\cup \hatL_{\sigma}\subseteq L_+(\sigma)$. 

The lemma thus follows. 
\end{proof}






We wish to compute the upper envelope $\calU(L_+(\sigma))$ for every cell $\sigma$ in the last cutting $\Xi_k$. To this end, we show in the following that $\calU(L_+(\sigma))$ for all cells $\sigma\in \Xi$ can be computed in $O(r^2\log n)$ time (precisely, the algorithm returns a binary search tree that represents $\calU(L_+(\sigma))$ for each cell $\sigma\in \Xi$). 

We work on the hierarchical cutting in a top-down manner. Suppose $\calU(L_+(\sigma'))$ for a cell $\sigma'$ has been computed (which is true initially when $\sigma'$ is the only cell of $\Xi$, in which case $L_+(\sigma')=\emptyset$ and thus $\calU(L_+(\sigma'))=\emptyset$). Then, for each child $\sigma$ of $\sigma'$, we compute $\calU(L_+(\sigma))$ as follows. By Lemma~\ref{lem:20}, $\calU(L_+(\sigma))$ is the upper envelope of $\calU(L_+(\sigma'))$ and $\calU(\hatL_{\sigma})$. The following Lemma~\ref{lem:30} provides an algorithm that computes $\calU(L_+(\sigma))$ based on $\calU(L_+(\sigma'))$ and $\calU(\hatL_{\sigma})$.
The technical crux of the result is an observation that the dual of $\calU(L_+(\sigma'))$, which is the lower hull of the dual points of the lines of $L_+(\sigma')$, has only $O(1)$ intersections with the dual of $\calU(\hatL_{\sigma})$, which is the lower hull of the dual points of the lines of $\hatL_{\sigma}$. Note that the algorithm does not report $\calU(L_+(\sigma))$ explicitly but rather returns a binary search tree representing $\calU(L_+(\sigma))$, which is obtained by splitting and merging binary search trees of $\calU(L_+(\sigma'))$ and $\calU(\hatL_{\sigma})$.

\begin{lemma}\label{lem:30}
$\calU(L_+(\sigma))$ can be obtained from $\calU(L_+(\sigma'))$ and $\calU(\hatL_{\sigma})$ in $O(\log n)$ time. 
\end{lemma}
\begin{proof}
Let $\calD(L_+(\sigma'))$ denote the set of dual points of the lines of $L_+(\sigma')$ and
$\calD(\hatL_{\sigma})$ the set of dual points of the lines of $\hatL_{\sigma}$.
Let $\calL(L_+(\sigma'))$ be the lower hull of the points of $\calD(L_+(\sigma'))$ and 
$\calL(\hatL_{\sigma})$ the lower hull of the points of $\calD(\hatL_{\sigma})$.
Note that $\calL(L_+(\sigma'))$ is dual to $\calU(L_+(\sigma'))$ while $\calL(\hatL_{\sigma})$ is dual to $\calU(\hatL_{\sigma})$. 

Our algorithm relies on the following {\em crucial observation:} $\calL(L_+(\sigma'))$ and $\calL(\hatL_{\sigma})$ cross each other at most $4$ times. In what follows, we first prove the observation and then discuss the algorithm. 

\paragraph{Proving the crucial observation.}
Let $Q$ be the set of vertices of the cell $\sigma'$. Let $Q^*$ be the set of dual lines of $Q$. As discussed in Section~\ref{sec:pre}, every cell of $\Xi$ is a (possibly unbounded) triangle. Hence, $|Q|\leq 3$. Since $\hatL_{\sigma}\subseteq L_{\sigma'}$, each line $\ell$ of $\hatL_{\sigma}$ must cross $\sigma'$. Therefore, $\ell$ is above at least one point of $Q$, and thus the dual point $\ell^*$ of $\ell$ must be below at least one dual line of $Q^*$. As such, all points of $\calD(\hatL_{\sigma})$ are below the upper envelope $\calU(Q^*)$ of $Q^*$. 

On the other hand, for each line $\ell$ of $L_+(\sigma')$, by definition, $\ell$ is completely below the cell $\sigma'$. Hence, the dual point $\ell^*$ of $\ell$ must be above all dual lines of $Q^*$. Therefore, all dual points of $\calD(L_+(\sigma'))$ are above the upper envelope $\calU(Q^*)$ of $Q^*$.

The above shows that $\calU(Q^*)$ separates $\calD(\hatL_{\sigma})$ from $\calD(\hatL_{\sigma})$: All points of   $\calD(\hatL_{\sigma})$ are below $\calU(Q^*)$ while all points of $\calD(L_+(\sigma'))$ are above $\calU(Q^*)$; see Fig.~\ref{fig:crucialobser}.

In what follows, we argue that the two lower hulls $\calL(L_+(\sigma'))$ and $\calL(\hatL_{\sigma})$ cross each other at no more than $4$ points. 

\begin{figure}[t]
\begin{minipage}[t]{\textwidth}
\begin{center}
\includegraphics[height=1.1in]{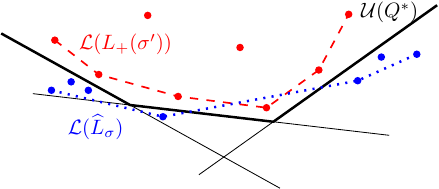}
\caption{\footnotesize Illustrating the proof of the crucial observation in Lemma~\ref{lem:30}. The three solid lines are those in $Q^*$. All blue points are those in $\calD(\hatL_{\sigma})$ while all red points are those in $\calD(L_+(\sigma'))$.}
\label{fig:crucialobser}
\end{center}
\end{minipage}
\vspace{-0.15in}
\end{figure}

Let $R$ be the region of the plane above $\calU(Q^*)$. Hence, all points of $\calD(L_+(\sigma'))$ are in $R$ and thus their lower hull $\calL(L_+(\sigma'))$ is completely inside $R$, implying that  $\calL(L_+(\sigma'))$ does not intersect $\calU(Q^*)$. 

On the other hand, since $|Q^*|\leq 3$, $\calU(Q^*)$ has at most two vertices. We claim that at most two edges of $\calL(\hatL_{\sigma})$ intersect $\calU(Q^*)$. Indeed, vertical lines through vertices of $\calU(Q^*)$ partition the plane into at most three slabs. Consider an edge $e$ connecting two adjacent vertices $u$ and $v$ of $\calL(\hatL_{\sigma})$. The edge $e$ intersects $\calU(Q^*)$ only if $u$ and $v$ are in different slabs. It is not difficult to see that the number of such edges is at most $2$. 

For any edge $e$ of $\calL(\hatL_{\sigma})$, if it does not intersect $\calU(Q^*)$, then $e$ must be completely below $\calU(Q^*)$ and thus cannot intersect $\calL(L_+(\sigma'))$. Hence, only edges of $\calL(\hatL_{\sigma})$ intersecting $\calU(Q^*)$ are possible to have intersections with $\calL(L_+(\sigma'))$.
Consider an edge $e$ of $\calL(\hatL_{\sigma})$ that intersects $\calU(Q^*)$. We argue that $e$ can cross $\calL(L_+(\sigma'))$ at most twice. 
Let $e'=e\cap R$. Clearly, among all points of $e$, only those on $e'$ can be possibly on $\calL(L_+(\sigma'))$ as $\calL(L_+(\sigma'))$ is inside $R$. Since $R$ is convex, $e'$ is a single line segment. Because $\calL(L_+(\sigma'))$ is convex, $e'$ crosses $\calL(L_+(\sigma'))$ at most twice. 
Since at most two edges of $\calL(\hatL_{\sigma})$ intersect $\calU(Q^*)$, 
$\calL(\hatL_{\sigma})$ and $\calL(L_+(\sigma'))$ cross each other at no more than $4$ points. 

This proves the crucial observation. 

\paragraph{The algorithm.}
Let $\calL$ be the lower hull of all points of $\calD(L_+(\sigma'))$ and
$\calD(\hatL_{\sigma})$. Note that $\calL$ is dual to $\calU(L_+(\sigma))$. Hence, to compute $\calU(L_+(\sigma))$, it suffices to compute $\calL$. By the virtue of the above crucial observation, we show that $\calL$ can be obtained in $O(\log n)$ time from $\calL(L_+(\sigma'))$ and $\calL(\hatL_{\sigma})$. 

Let $A$ be the set of at most four intersection points between  $\calL(L_+(\sigma'))$ and $\calL(\hatL_{\sigma})$. We first compute $A$. After having $\calU(Q^*)$, which takes $O(1)$ time to compute, for each vertex $v$ of $\calU(Q^*)$, we find the edge $e$ (if exists) of $\calL(\hatL_{\sigma})$ whose left endpoint is left of $v$ and right endpoint is right of $v$. The edge $e$ can be determined in $O(\log n)$ time using the binary search tree of $\calL(\hatL_{\sigma})$. Then, we compute the at most two intersections (if any) between $e$ and $\calL(L_+(\sigma'))$, which again can be done in $O(\log n)$ time. 
In this way, $A$ can be computed in $O(\log n)$ time since $\calU(Q^*)$ has at most two vertices. 

Vertical lines through the points of $A$ partition the plane into at most five slabs. For each slab $B$, the portion of $\calL(\hatL_{\sigma})$ in $B$ is either completely above or completely below the portion of $\calL(L_+(\sigma'))$ in $B$. Hence, the lower hull of the points of $\calD(L_+(\sigma'))\cup \calD(\hatL_{\sigma})$ in each slab $B$ is already available. More specifically, we split the binary search tree for $\calL(\hatL_{\sigma})$ at the $x$-coordinates of points of $A$ and obtain at most five trees representing portions of $\calL(\hatL_{\sigma})$ in the slabs, respectively. We do the same for $\calL(L_+(\sigma'))$. As such, we have a binary search tree representing the lower hull of the points of $\calD(L_+(\sigma'))\cup \calD(\hatL_{\sigma})$ in each slab $B$. The lower hull $\calL$ of all points can thus be obtained by merging these lower hulls in the slabs by computing their lower common tangents. Specifically, we can start with merging the lower hulls of the leftmost two slabs by computing their lower common tangent, which can be accomplished in $O(\log n)$ time~\cite{ref:OvermarsMa81} since the two lower hulls are separated by a vertical line. After the leftmost two hulls are merged, we continue to merge the third one, and so on. In this way, $\calL$ can be computed in $O(\log n)$ time. 
\end{proof}

By virtue of Lemma~\ref{lem:30}, we can compute $\calU(L_+(\sigma))$ for all cells $\sigma\in \Xi$ in $O(|\Xi|\cdot \log n)$ time, which is $O(r^2\log n)$ as $\Xi$ has $O(r^2)$ cells. 

Consider a cell $\sigma\in \Xi_{k}$. 
For any point $p\in P(\sigma)$, let $L^+_{\sigma}(p)$ denote the subset of lines of $L_{\sigma}$ below $p$. It is easy to see that $L_+(p)=L_+(\sigma)\cup L^+_{\sigma}(p)$. Hence, $\calU(L_+(p))$ is the lower envelope of $\calU(L_+(\sigma))$ and $\calU(L^+_{\sigma}(p))$. The above already computes $\calU(L_+(\sigma))$. Suppose $\calU(L^+_{\sigma}(p))$ is also available. Then, we compute $\calU(L_+(p))$ by merging $\calU(L_+(\sigma))$ and $\calU(L^+_{\sigma}(p))$ in $O(\log n)$ time by the algorithm of Lemma~\ref{lem:30}. Indeed, since all lines of $L^+_{\sigma}(p)$ cross $\sigma$ while all lines of $L_+(\sigma)$ are completely below $\sigma$, the same algorithm as Lemma 4 is also applicable here. Consequently, once the upper envelopes $\calU(L^+_{\sigma}(p))$ for all points $p\in P$ are computed, $\calU(L_+(p))$ for all $p\in P$ can be computed in additional $O(m\log n)$ time. 

It remains to compute $\calU(L^+_{\sigma}(p))$. For this, we recursively apply the above algorithm on $L_{\sigma}$ and $P(\sigma)$. A subtle issue is that while $|L_{\sigma}|\leq n/r$, we do not have a good upper bound for $P(\sigma)$. To address this issue, we do the following. If $|P(\sigma)|>m/r^2$, then we arbitrarily partition $P(\sigma)$ into groups of at most $m/r^2$ points each. Since $\Xi_k$ has $O(r^2)$ cells, the number of groups thus created for all cells $\sigma\in \Xi_k$ is still bounded by $O(r^2)$. Now for each cell $\sigma\in \Xi_k$, for each group $P'$ of $P(\sigma)$, we apply the above algorithm recursively on $P'$ and $L_{\sigma}$. As such, we obtain the following recurrence for the runtime of the entire algorithm excluding the time for presorting $L$, for any $1\leq r\leq n$: 
\begin{equation}\label{equ:10}
    T(m,n)=O(nr + m\log n+r^2\log n)+O(r^2)\cdot T(m/r^2,n/r).
\end{equation}

Instead of solving the recurrence now, we will use the recurrence later in our combined algorithm in Section~\ref{sec:combine}.

\section{The second algorithm in the dual plane} 
\label{sec:dual}
In this section, we discuss another algorithm, which deals with the problem in the dual plane. It mostly follows Wang's algorithm in \cite{ref:WangCo23} with a slight change at the very end of this section (more specifically, we make Wang's algorithm recursive). 

Let $P^*$ be the set of lines dual to the points of $P$ and $L^*$ the set of points dual to the lines of $L$.
Consider a point $p\in P$. Recall that $F_p(L)$ is the face of $\calA(L)$ that contains $p$. 
Without loss of generality, we assume that the dual line $p^*$ of $p$ is horizontal. 
In the dual plane, $p^*$ divides the points of $L^*$ into two subsets and the portions of the convex hulls of the two subsets between their inner common tangents are dual to $F_p(L)$~\cite{ref:AgarwalPa902,ref:EdelsbrunnerIm89}; see Fig~\ref{fig:dual}.
Specifically, define $L^*_+(p^*)$ (resp., $L^*_-(p^*)$) to be the subset of points of $L^*$ above (resp., below) $p^*$. Let $H_+(p^*)$  be the lower hull of the convex hull of $L^*_+(p^*)$ and $H_-(p^*)$ the upper hull of $L^*_-(p^*)$; see Fig~\ref{fig:dual}. 
Hence, the boundary of $F_p(L)$ is dual to the portions of $H_+(p^*)$ and $H_-(p^*)$ between their inner common
tangents; let $F^*_p(L)$ denote the dual of $F_p(L)$.

\begin{figure}[t]
\begin{minipage}[t]{\textwidth}
\begin{center}
\includegraphics[height=1.5in]{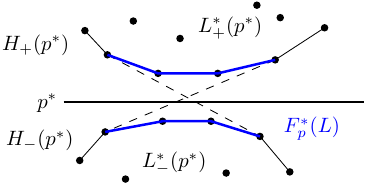}
\caption{\footnotesize Illustrating the notation in the dual plane: $F^*_p(L)$, which is dual to $F_p(L)$, is composed of the blue edges between the two inner common tangents (the dashed segments).}
\label{fig:dual}
\end{center}
\end{minipage}
\vspace{-0.15in}
\end{figure}

The algorithm will compute binary search trees of height $O(\log n)$ to represent 
$H_+(p^*)$ and $H_-(p^*)$, respectively. Using the trees, their inner common tangents can be computed in $O(\log m)$ time~\cite{ref:GuibasCo91} and then $F^*_p(L)$ can be reported in additional $O(|F^*_p(L)|)$ time. Below we only discuss how to compute $H_+(p^*)$ (i.e., compute the tree for it). 

The algorithm starts to compute a hierarchical $(1/r)$-cutting $\Xi_0,\Xi_1,\ldots,\Xi_k$
for the lines of $P^*$, with $r\in [1,m]$ to be determined later. Let $\Xi$ denote the set of all cells of $\Xi_i$, for all $0\leq i\leq k$. We also keep the conflict lists $P^*_\sigma$
for all cells $\sigma\in \Xi$. All of these can be done in $O(mr)$ time~\cite{ref:ChazelleCu93}. As discussed in Section~\ref{sec:pre}, $k=O(\log r)$.
We next compute the set $L^*(\sigma)$, which is defined to be $L^*\cap \sigma$, for all cells $\sigma\in \Xi$. This can be done in $O(n\log r)$ time by point locations following the hierarchical cutting in the same way as discussed in Section~\ref{sec:primal}. 
If we sort points of $L^*$ by $x$-coordinate initially, then we can obtain the sorted lists of $L^*(\sigma)$ of all cells $\sigma\in \Xi$ in $O(n\log r)$ time. Indeed, if we perform point locations for points of $L^*$ following their sorted order, then points of each $L^*(\sigma)$ are automatically sorted.
Using the sorted lists, for each cell $\sigma\in \Xi$, we compute the convex hull of $L^*(\sigma)$ in linear time and store it in a (balanced) binary search tree. All above takes $O(mr+n\log r)$ time (excluding the $O(n\log n)$ time for presorting $L^*$).

In addition, for each cell $\sigma$ of the last cutting $\Xi_k$, if $|L^*(\sigma)|>n/r^2$, then we further partition $\sigma$ into $\Theta(|L^*(\sigma)|\cdot r^2/n)$ cells each containing at most $n/r^2$ points of $L^*$. This can be easily done in $O(n)$ time for all cells of $\Xi_k$ since points of each $L^*(\sigma)$ are already sorted.
For convenience, we use $\Xi_{k+1}$ to refer to the set of all new cells obtained above. Note that $\Xi_{k+1}$ has $O(r^2)$ cells. For each cell $\sigma'\in \Xi_{k+1}$, if $\sigma'$ is in the cell $\sigma$ of $\Xi_k$, then we also say that $\sigma$ is the {\em parent} of $\sigma'$ and $\sigma'$ is a {\em child} of $\sigma$. 
We also define $L^*(\sigma')=L^*\cap \sigma'$, and maitain the convex hull of $L^*(\sigma')$ for each cell $\sigma'\in \Xi_{k+1}$. This takes $O(n)$ time for all cells $\sigma'$ of $\Xi_{k+1}$ as $L^*$ is already sorted.


The algorithm next processes each line $p^*\in P^*$ as follows. Without loss of generality, we assume that $p^*$ is horizontal. Denote by $\Psi(p^*)$ the set of  cells $\sigma$ of $\Xi_i$ crossed by $p^*$, for all $0\leq i\leq k$. Define $\Psi_{k+1}(p^*)$ to be the set of cells $\sigma$ of $\Xi_{k+1}$ crossed by $p^*$. For each cell $\sigma\in \Psi_{k+1}(p^*)$, denote by $\sigma_+(p^*)$ the portion of $\sigma$ above $p^*$, and let $L^*(\sigma_+(p^*))=L^*\cap \sigma_+(p^*)$. 


Define $\Psi_+(p^*)=\{\sigma \ | \ \text{$\sigma$'s parent is in $\Psi(p^*)$ and $\sigma$ is completely above the line $p^*$}\}$.

It is proved in \cite{ref:WangCo23} that 
$L^*_+(p^*)$ is the union of $\bigcup_{\sigma\in \Psi_+(p^*)}L^*(\sigma)$ and $\bigcup_{\sigma\in \Psi_{k+1}(p^*)}L^*(\sigma_+(p^*))$.
As such, if we have convex hulls of $L^*(\sigma)$ for all cells $\sigma\in \Psi_+(p^*)$ and convex hulls of $L^*(\sigma_+(p^*))$ for all cells $\sigma\in \Psi_{k+1}(p^*)$, then $H_+(p^*)$ is the lower hull of the vertices of all these convex hulls. 
Define $\calH^1_+(p^*)$ to be the set of convex hulls of $L^*(\sigma)$ for all cells $\sigma\in \Psi_+(p^*)$ and $\calH^2_+(p^*)$ the set of convex hulls of $L^*(\sigma_+(p^*))$ for all cells $\sigma\in \Psi_{k+1}(p^*)$. Let $\calH_+(p^*)=\calH^1_+(p^*)\cup \calH^2_+(p^*)$. Thus, $H_+(p^*)$ is the lower hull of the vertices of all convex hulls of $\calH_+(p^*)$. 

It has been proved in \cite{ref:WangCo23} that $\sum_{p^*\in P^*}|\calH_+(p^*)|=O(mr)$ and convex hulls of $\calH_+(p^*)$ are pairwise disjoint. Using this property, once we have (binary search trees for representing) convex hulls of $\calH_+(p^*)$, $H_+(p^*)$ can be computed in additional $O(mr\log n)$ time~\cite{ref:WangCo23} (details about this step will be further discussed in Section~\ref{sec:combine}). 
Thanks to our preprocessing above on the hierarchical cutting, 
convex hulls of $\calH^1_+(p^*)$ can be obtained in $O(mr)$ time. 
For computing the convex hulls of $\calH^2_+(p^*)$, Wang's algorithm~\cite{ref:WangCo23} does so by brute force, which takes $O(mn/r)$ time. 
As such, the total time of the algorithm is $O(n\log n+n\log r+ mr\log n+mn/r)$ (in particular, the $O(n\log n)$ factor is for presorting $L^*$). Setting $r=\sqrt{n/\log n}$ leads to the $O(n\log n+m\sqrt{n\log n})$ time complexity given in \cite{ref:WangCo23}.

\paragraph{Modifying Wang's algorithm.} We modidfy Wang's algorithm as follows. Instead of computing the convex hulls of $\calH^2_+(p^*)$ by brute force, we compute them recursively. 
Each subproblem, which is on $P^*_{\sigma}$ and $L^*(\sigma)$ for each cell $\sigma\in \Xi_{k+1}$, is the following: Compute the lower hull of $L^*(\sigma_+(p^*))$ for every line $p^*\in P^*_{\sigma}$. As $\Xi_{k+1}$ has $O(r^2)$ cells, $|P^*_{\sigma}|\leq m/r$, and $|L^*(\sigma)|\leq n/r^2$, we obtain the following recurrence for the runtime of the whole algorithm (excluding the time for presorting $L^*$), for any $1\leq r\leq m$:
\begin{equation}\label{equ:20}
    T(m,n)=O(mr\log n+n\log r)+O(r^2)\cdot T(m/r,n/r^2).
\end{equation}

We will use the recurrence in our combined algorithm in Section~\ref{sec:combine}.

\paragraph{Remark.} 
The term $O(mr\log n)$ in \eqref{equ:20} is due to the procedure of computing the lower hull $H_+(p^*)$ by merging the convex hulls of $\calH_+(p^*)$, for all $p^*\in P^*$. Recall that $\sum_{p^*\in P^*}|\calH_+(p^*)|=O(mr)$. If this procedure could be performed in $O(mr)$ time, then the term in \eqref{equ:20} would become $O(mr)$. 


\section{Combining the two algorithms}
\label{sec:combine}
In this section, we combine the two algorithms presented in the last two sections to obtain a final algorithm to compute all non-empty faces of $\calA(L)$ for the points of $P$. 

We first discuss the symmetric case where $m=n$. If we apply~\eqref{equ:20} and then \eqref{equ:10} using the same $r$, we can obtain the following recurrence:
\begin{equation*}
    T(n,n)=O(nr\log n + r^4\log (n/r^2))+O(r^4)\cdot T(n/r^3,n/r^3).
\end{equation*}
Setting $r=n^{1/3}/\log n$ leads to the following
\begin{equation}\label{equ:30}
    T(n,n)=O(n^{4/3})+O((n/\log^3 n)^{4/3})\cdot T(\log^3n,\log^3n).
\end{equation}
The recurrence solves to $T(n,n)=n^{4/3}\cdot 2^{O(\log^*n)}$.
In the following, we improve the algorithm to $O(n^{4/3})$ time using the $\Gamma$-algorithm framework of Chan and Zheng~\cite{ref:ChanHo23}.

\subsection{Improvement}

At the outset, we apply the recurrence \eqref{equ:40} one more time to obtain the following:
\begin{align}\label{equ:40}
  T(n,n)= O(n^{4/3}) + O((n/b)^{4/3})\cdot T(b,b),
\end{align}
where $b=(\log\log n)^3$.

As $b$ is very small, we show that after $O(n)$ time preprocessing, we can solve each subproblem $T(b,b)$ in
$O(b^{4/3})$ time. 
By integrating this outcome into \eqref{equ:40}, we derive the result
$T(n,n)=O(n^{4/3})$.

More precisely, we will demonstrate that following a preprocessing step requiring $O(2^{\text{poly}(b)})$ time, where $\text{poly}(\cdot)$ represents a polynomial function, each subproblem $T(b,b)$ can be solved by performing only $O(b^{4/3})$ comparisons. Alternatively, we can solve $T(b,b)$ using an algebraic decision tree with a height of $O(b^{4/3})$. Given that $b = (\log\log n)^3$, the term $2^{\text{poly}(b)}$ remains bounded by $O(n)$. To adapt this approach to the conventional computational model (such as the standard real-RAM model), we explicitly construct the decision tree for the aforementioned algorithm. This construction, which can also be viewed as part of the preprocessing for solving $T(b,b)$, can be accomplished in $O(2^{\text{poly}(b)})$ time. Consequently, with just $O(n)$ time spent on preprocessing, we can efficiently solve each subproblem $T(b,b)$ in $O(b^{4/3})$ time. In the subsequent discussion, for notational convenience, we will use $n$ instead of $b$. Our objective is to establish the following lemma.
\begin{lemma}\label{lem:40}
After $O(2^{\poly(n)})$ time preprocessing, $T(n,n)$
can be solved using $O(n^{4/3})$ comparisons.
\end{lemma}

In the following, we prove Lemma~\ref{lem:40}. 

We apply recurrence~\eqref{equ:20} by setting $m=n$ and $r=n^{1/3}$, and obtain the following
\begin{align}\label{equ:50}
  T(n,n)= O(n^{4/3}\log n)+O(n^{2/3})\cdot T(n^{2/3},n^{1/3}).
\end{align}
As remarked at the end of Section~\ref{sec:dual}, the $O(n^{4/3}\log n)$ term is due to the procedure of computing the lower hull $H_+(p^*)$ by merging the convex hulls of $\calH_+(p^*)$, for all $p^*\in P^*$. If the procedure could be performed using $O(n^{4/3})$ comparisons, then the term would become $O(n^{4/3})$. Note that $\sum_{p^*\in P^*}|\calH_+(p^*)|=O(n^{4/3})$. As such, to solve $T(n,n)$ by $O(n^{4/3})$ comparisons, there are two challenges: (1) perform the merging procedure on the $O(n^{4/3})$ convex hulls using $O(n^{4/3})$ comparisons; (2) solve each subproblem $T(n^{2/3},n^{1/3})$ using $O(n^{2/3})$ ``amortized'' comparisons, so that all $O(n^{2/3})$ such subproblems in \eqref{equ:50} together cost $O(n^{4/3})$ comparisons.

\paragraph{$\boldsymbol{\Gamma}$-algorithm framework.}

To address these challenges, we resort to a framework known as the {\em $\Gamma$-algorithm} for bounding decision tree complexities, as introduced by Chan and Zheng in their recent work~\cite{ref:ChanHo23}. We provide a brief overview here, with more detailed information available in~\cite[Section 4.1]{ref:ChanHo23}.

In essence, this framework constitutes an algorithm that exclusively counts the number of comparisons, termed {\em $\Gamma$-comparisons} in~\cite{ref:ChanHo23}, to determine if a point belongs to a semialgebraic set of degree $O(1)$ within a constant-dimensional space. Solving our problem of finding the non-empty faces is equivalent to locating the cell $C^*$ that contains a point $p^*$ defined by the input (i.e., the lines of $L$ and the points of $P$ in our problem) within an arrangement $\calA^*$ comprising the boundaries of $O(n)$ semialgebraic sets in an $O(n)$-dimensional space (because our input size is $O(n)$). Constructing this arrangement can be achieved in $O(2^{\poly(n)})$ time without inspecting the input values, thereby obviating the need for any comparisons. Notably, the number of cells in $\calA^*$ is bounded by $n^{O(n)}$.

As the $\Gamma$-algorithm progresses, it maintains a set $\Pi$ consisting of cells from $\calA^*$. Initially, $\Pi$ comprises all cells of $\calA^*$. During the algorithm's execution, $\Pi$ can only shrink, yet it always contains the sought-after cell $C^*$. Upon completion of the algorithm, $C^*$ is located.

We define a potential function $\Phi$ as $\log |\Pi|$. Given that $\calA^*$ contains $n^{O(n)}$ cells, initially $\Phi$ is $O(n\log n)$. For any operation or subroutine of the algorithm, let $\Delta\Phi$ denote the change in $\Phi$ (i.e., the value of $\Phi$ after the operation minus its value before the operation). Since $\Phi$ monotonically decreases throughout the algorithm, we always have $\Delta\Phi \leq 0$. The sum of $-\Delta\Phi$ over the entire algorithm is $O(n\log n)$. Consequently, this allows us to accommodate costly operations or subroutines during the algorithm, as long as they result in a decrease in $\Phi$.

Two algorithmic tools are introduced in \cite{ref:ChanHo23} under the framework: the {\em basic search lemma}~\cite[Lemma 4.1]{ref:ChanHo23} and the {\em search lemma}~\cite[Lemma A.1]{ref:ChanHo23}. Roughly speaking, these lemmas operate as follows: when presented with a set of $r$ predicates, 
where each predicate assesses whether $\gamma(x)$ holds true for the input vector $x$, it is guaranteed that at least one of these predicates is true for all inputs within the active cells. In such cases, the basic search lemma can identify a predicate that holds true by conducting $O(1-r\cdot \Delta\Phi)$ comparisons.
The search lemma is for scenarios involving a binary tree (or a more general directed acyclic graph of $O(1)$ degree) with nodes $v$, each associated with a predicate $\gamma_v$, and where for each internal node $v$, $\gamma_v$ implies $\gamma_u$ for a child $u$ of $v$ across all inputs in the active cells. This lemma allows the identification of a leaf node $v$ for which $\gamma_v$ holds true with $O(1-\Delta\Phi)$ comparisons.

As discussed in \cite{ref:ChanHo23}, intuitively the basic search lemma provides a mild form of nondeterminism, allowing us to ``guess'' which one of $r$ choices is correct, with only $O(1)$ amortized cost instead of $O(r)$. This situation naturally occurs in the context of point location, where we seek to determine which one of the $r$ cells contains a given point. Another noteworthy application of both lemmas, as also discussed in \cite{ref:ChanHo23}, is the task of finding the predecessor of a query number within a sorted list of input numbers. As will be seen later in our algorithm, some subproblems requiring the $\Gamma$-algorithm framework also involves point locations as well as 
finding predecessors within sorted lists, making both the basic search lemma and the search lemma highly applicable. 

\medskip

In the following two subsections, we will address the aforementioned challenges individually. By slightly abusing the notation, let $P$ represent a set of $n$ points, and $L$ denote a set of $n$ lines, for the problem in Recurrence~\eqref{equ:50}.

\subsection{The convex hull merge procedure}
For each dual line $p^*\in P^*$ (without loss of generality, assume it is horizontal), we have a set of convex hulls $\calH_+(p^*)$ such that (1) they are completely above $p^*$; (2) they are pairwise disjoint; (3) $H_+(p^*)$ is the lower hull of them; (4) $\sum_{p^*\in P^*}|\calH_+(p^*)|=O(n^{4/3})$. Again, for each convex hull of $\calH_+(p^*)$, we have a binary search tree that represents it. Our objective is to compute (a binary search tree that represents) the lower hull (which is $H_+(p^*)$) of the vertices of all convex hulls of $\calH_+(p^*)$, for all $p^*\in P^*$.

Define $K_{p^*}=|\calH_+(p^*)|$. Hence, $\sum_{p^*\in P^*}K_{p^*}=O(n^{4/3})$. Wang~\cite{ref:WangCo23} gave an algorithm that can constructs $H_+(p^*)$ in $O(K_{p^*}\log n)$ time, resulting in a total of $O(n^{4/3}\log n)$ time for all $p^*\in P^*$. We will convert the algorithm to a more efficient $\Gamma$-algorithm that uses only $O(n^{4/3})$ comparisons. In what follows, we first review Wang's algorithm. 

\paragraph{Constructing $\boldsymbol{H_+(p^*)}$: A conventional algorithm from \cite{ref:WangCo23}.}
Since $H_+(p^*)$ is the lower hull of all convex hulls of $\calH_+(p^*)$, 
it is only necessary to focus on the lower hull of each convex hull of $\calH_+(p^*)$. For each convex hull of $\calH_+(p^*)$, since we have its binary search tree, we can obtain a binary search tree only 
representing its lower hull in $O(\log n)$ time. This is achieved by initially identifying the leftmost and rightmost vertices of the convex hulls and subsequently executing split/merge operations on these trees. 
The total time for doing this for all convex hulls of $\calH_+(p^*)$ is $O(K_{p^*}\log n)$. 

The next step entails calculating the portions of each lower hull $H$ of $\calH_+(p^*)$ that are vertically visible from the line $p^*$. A point $q\in H$ is {\em vertically visible} from $p^*$ if the vertical line segment connecting $q$ to $p^*$ does not intersect any other lower hull of $\calH_+(p^*)$. Remarkably, these visible portions collectively form the lower envelope of $\calH_+(p^*)$, denoted by $\calL(\calH_+(p^*))$; see Fig.~\ref{fig:lowenvelope}.
Next we compute $\calL(\calH_+(p^*))$ in $O(K_{p^*}\log n)$ time.

\begin{figure}[t]
\begin{minipage}[t]{\textwidth}
\begin{center}
\includegraphics[height=1.3in]{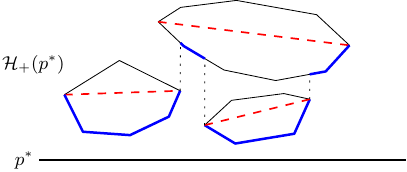}
\caption{\footnotesize Illustrating the lower envelope $\calL(\calH_+(p^*))$ (the thick blue edges) of the convex hulls of $\calH_+(p^*)$. The dashed segments inside convex hulls are their representative segments.}
\label{fig:lowenvelope}
\end{center}
\end{minipage}
\vspace{-0.15in}
\end{figure}


For each convex hull $H$ within $\calH_+(p^*)$, we define its {\em representative segment} as the line segment connecting the leftmost and rightmost endpoints of $H$ (see Fig.~\ref{fig:lowenvelope}). Define $S$ to be the set of the representative segments of all convex hulls of $\calH_+(p^*)$. Because the convex hulls are pairwise disjoint, so are the segments of $S$. Observe that 
the lower envelope $\calL(S)$ of $S$ corresponds to $\calL(\calH_+(p^*))$ in the following manner. For each maximal segment  $\overline{ab}$ of $\calL(S)$, suppose it lies on the representative segment of a convex hull $H$ of $\calH_+(p^*)$. Then the vertical projection of $\overline{ab}$ onto the lower hull of $H$ constitutes a maximal portion of the lower hull of $H$ on $\calL(\calH_+(p^*))$. This specific portion can be obtained in $O(\log n)$ time by splitting the binary search tree representing the lower hull of $H$ at the $x$-coordinates of $a$ and $b$, respectively.
Consequently, once $\calL(S)$ is available, $\calL(\calH_+(p^*))$, where each maximal portion is represented by a binary search tree, can be obtained in additional $O(|\calL(S)|\log n)$ time. Since segments of $S$ are disjoint, we have $|\calL(S)|\leq 2K_{p^*}-1$ and constructing $\calL(S)$ can be accomplished in $O(K_{p^*}\log K_{p^*})$ time by a straightforward plane sweeping algorithm (see Section~\ref{sec:second} for details). Note that $K_{p^*}\leq n$ since vertices of each convex hull of $\calH_+(p^*)$ are all from $L^*$ (whose size is $n$).
As such, $\calL(\calH_+(p^*))$ can be computed in $O(K_{p^*}\log n)$ time.

With $\calL(\calH_+(p^*))$, we proceed to compute the lower hull $H_+(p^*)$ in additional $O(K_{p^*}\log n)$ time, as follows. As previously discussed, $\calL(\calH_+(p^*))$ comprises at most $2K_{p^*}-1$ pieces, organized from left to right. Each piece is a portion of a lower hull of $\calH_+(p^*)$ and is represented by a binary search tree.

To begin, we merge the first two pieces by computing their lower common tangent, a task achievable in $O(\log n)$ time~\cite{ref:OvermarsMa81}, since these two pieces are separated by a vertical line. Following this merge operation, we obtain a binary search tree that represents the lower hull of the first two pieces of $\calL(\calH_+(p^*))$. Subsequently, we repeat this merging process, proceeding to merge this lower hull with the third piece of $\calL(\calH_+(p^*))$ in a similar manner. This process continues until all pieces of $\calL(\calH_+(p^*))$ have been merged, culminating in a binary search tree that represents $H_+(p^*)$.
The runtime for this procedure is bounded by $O(K_{p^*}\log n)$, as each merge operation consumes $O(\log n)$ time, and $\calL(\calH_+(p^*))$ contains at most $2K_{p^*}-1$ pieces.

As such, the lower hull $H_+(p^*)$ can be computed in $O(K_{p^*}\log n)$ time.

Applying the above algorithm to all dual lines $p^*$ of $P^*$ will compute the lower hulls $H_+(p^*)$ for all $p^*\in P^*$. It should be noted that after the algorithm is applied to one line $p^*\in P^*$, binary search trees of convex hulls of $\calH_+(p^*)$ may have been destroyed due to the split and merge operations during the algorithm. The destroyed convex hulls may be used later for other lines of $P^*$. To address the issue, we can use persistent binary search trees with path-copying~\cite{ref:DriscollMa89,ref:SarnakPl86} to represent convex hulls so that standard operations on the trees (e.g., merge, split) can be performed in $O(\log n)$ time each and after each operation the original trees are still kept. As such, whenever the algorithm is invoked for a new line of $P^*$, we can always access the original trees representing the convex hulls, and thus the runtime of the algorithm is not affected. 

\paragraph{Constructing $\boldsymbol{H_+(p^*)}$: a new and faster $\boldsymbol{\Gamma}$-algorithm.}
We now convert the above algorithm to a faster $\Gamma$-algorithm. In the above algorithm, the procedures that take $O(K_{p^*}\log n)$ time are the following: (1) Computing the trees representing the lower hulls of $\calH_+(p^*)$; (2) computing the lower envelope $\calL(S)$ of $S$; (3) computing the lower envelope $\calL(\calH_+(p^*))$ of $\calH_+(p^*)$ by using $\calL(S)$; (4) computing $H_+(p^*)$ by merging the pieces of $\calL(\calH_+(p^*))$. Each of these procedures leads to an overall $O(n^{4/3}\log n)$ time for all $p^*\in P^*$ since $\sum_{p^*\in P^*}|\calH_+(p^*)|=O(n^{4/3})$. 

For each of these procedures, we will design a corresponding ``$\Gamma$-procedure'' that uses only $O(n^{4/3})$ comparisons for all $p^*\in P^*$ (thus the amortized cost for each $p^*$ is $O(K_{p^*})$). In particular, the fourth procedure poses the most challenge and its solution is also the most interesting. 

\subsubsection{The first procedure}
For the first procedure, let $H$ be a convex hull of $\calH_+(p^*)$. Given a binary search tree $T_H$ representing $H$, the objective is to obtain a binary search tree from $T_H$ to represent the lower hull of $H$. The first step is to find the leftmost and rightmost vertices of $H$. A straightforward approach, as in Wang's approach, is to do binary search on $T_H$, which takes $O(\log n)$ time. This results in a total of $O(n^{4/3}\log n)$ time for all convex hulls of $\calH_+(p^*)$ and for all $p^*\in P^*$. 
In the following, we propose a new approach that solve the problem in $O(n^{4/3})$ time. In fact, this is a conventional algorithm (not a particular $\Gamma$-algorithm). 

Recall that there are $O(r^2)$ convex hulls in total because each convex hull is for the points of $L^*(\sigma)$ for each cell $\sigma$ in our hierarchical $(1/r)$-cutting $\Xi$ for the dual lines of $P^*$ and $\Xi$ has $O(r^2)$ cells, with $r=n^{1/3}$. Also, recall (from Section~\ref{sec:dual}) that the total number of points of $L^*(\sigma)$ in all cells $\sigma\in \Xi$ is bounded by $O(n\log r)$ since each point of $L^*$ is in at most $O(\log r)$ cells. Therefore, the total number of vertices of all convex hulls involved in our algorithm is $O(n\log n)$.

\begin{figure}[t]
\begin{minipage}[t]{\textwidth}
\begin{center}
\includegraphics[height=1.0in]{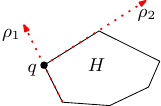}
\caption{\footnotesize Illustrating the two rays $\rho_1$ and $\rho_2$.}
\label{fig:rays}
\end{center}
\end{minipage}
\vspace{-0.15in}
\end{figure}

Let $q$ be the leftmost vertex of a convex
hull $H\in \calH_+(p^*)$ for a dual line $p^*\in P^*$ (assuming that $p^*$ is horizontal). Let
$\rho_1$ be the ray from the counterclockwise neighboring vertex of $q$ to $q$, i.e., $\rho_2$ contains
the counterclockwise edge of $H$ incident to $q$ (see Fig.~\ref{fig:rays}). Let $\rho_2$ be the ray from $q$
to the clockwise neighboring vertex of $q$ on $H$, i.e., $\rho_2$ contains
the clockwise edge of $H$ incident to $q$. Observe that for another dual line $p^*_1\in P^*$, $q$ is still the leftmost vertex of $H$ with respect to $p^*_1$ as
long as the direction perpendicular to $p_1^*$ is within the angle from $\rho_1$
clockwise to $\rho_2$.
Based on this observation, before we process any dual line $p^*$ of $P^*$ (i.e., computing their $H_+(p^*)$), we sort the perpendicular
directions of all lines of $P^*$ along with the directions of all edges of all
convex hulls for all cells of $\Xi$, which takes $O(n\log^2 n)$ time as the total size of all these convex hulls is $O(n\log n)$ as discussed above. 
Let $A$ be the sorted list. We process lines of $P^*$ following their order in $A$.
Let $p^*$ and $p^*_1$ be two consecutive lines of $P^*$ in $A$. After we process $p^*$, the directions between $p^*$ and $p^*_1$ in $A$ correspond to those
convex hulls whose leftmost vertices need to get updated, and we then
update them before we process $p^*_1$. The
total time for updating the convex hulls as above is proportional to the total
size of all convex hulls, which is $O(n\log n)$.

In this way, whenever we process a line $p^*$, for each convex hull $H\in \calH_+(p^*)$, we have its leftmost vertex available, and more specifically, we have a pointer pointing to the leftmost vertex in the binary search tree of $H$. 

In a similar way, we can also find the rightmost vertices of all convex hulls of $\calH_+(p^*)$. As such, finding the leftmost and rightmost vertices of every convex hull $H\in \calH_+(p^*)$ for all $p^*\in P^*$ can be done in $O(n^{4/3})$ time by a conventional algorithm. 

After having the leftmost and rightmost vertices of each convex hull $H$, we can split/merge the tree $T_H$ using the two vertices. This takes $O(\log n)$ time in a conventional algorithm. However, these operations do not involve comparisons and thus do not introduce any cost in a $\Gamma$-algorithm. 

In summary, we have a $\Gamma$-procedure that can solve the first procedure using $O(n^{4/3})$ comparisons for all $p^*\in P^*$. 

\subsubsection{The second procedure}
\label{sec:second}
The second procedure is to compute the lower envelope $\calL(S)$ of $S$. First of all, since the leftmost and rightmost points of each convex hull $H$ of $\calH_+(p^*)$ have already been computed in the first procedure, we have $S$ available as well. Let $t=|S|$. Recall that $t\leq 2K_{p^*}-1$.
The lower envelope $\calL(S)$ of $S$ can be easily computed in $O(t\log t)$ time by a standard plane sweeping algorithm, which will result in a total of $O(n^{4/3}\log n)$ time for all $p^*\in P^*$. 
Our goal is to design a $\Gamma$-algorithm that uses only $O(n^{4/3})$ comparisons. 

We first briefly discuss the plane sweeping algorithm. Sort endpoints of all segments of $S$ from left to right. Sweep a vertical line $\ell$ from left to right in the plane. We maintain a subset $S(\ell)\subseteq S$ of segments that intersect $\ell$ by a binary search tree $T_{\ell}$ in which segments are ordered by the $y$-coordinates of their intersections with $\ell$. When $\ell$ hits the left (resp., right) endpoint of a segment $s\in S$, we insert $s$ into $T_{\ell}$ (resp., delete $s$ from $T_{\ell}$). In either event, we update the lowest segment of $T_{\ell}$, which is on $\calL(S)$; the lowest segment is stored at the leftmost leaf of $T_{\ell}$. Each event can be processed in $O(\log t)$ time and thus the total time of the algorithm is $O(t\log t)$. 

We next convert the algorithm to a $\Gamma$-procedure. First of all, the sorting can be done using $O(t-\Delta\Phi)$ comparisons (e.g., using the search lemma, as already discussed in \cite{ref:ChanHo23}; more specifically, we can find a predecessor of a value among a list of $t'$ sorted values using $O(1-\Delta\Phi)$ comparisons, by applying the search lemma to a binary search tree of $t'$ leaves).

In the line sweeping procedure, we maintain a tree $T_{\ell}$ that involves insertion/deletion operations. Each insertion/deletion essentially is to find the predecessor of a list of values in $T_{\ell}$. By applying the search lemma, as discussed above, each operation can be done using $O(1-\Delta\Phi)$ comparisons (note that once the predecessor is found, adjusting the tree does not cost any comparisons). As such, the sweeping procedure can be performed using $O(t-\Delta\Phi)$ comparisons. 

The above shows that the second procedure can be accomplished using $O(t-\Delta\Phi)$ comparisons for each $p^*\in P^*$. Since $t\leq 2K_{p^*}-1$, $\sum_{p^*\in P^*}K_{p^*}=O(n^{4/3})$, and the sum of $-\Delta\Phi$ in the entire algorithm is $O(n\log n)$,  we obtain that the total number of comparisons needed in the second procedure for all $p^*\in P^*$ is bounded by $O(n^{4/3})$. 

\subsubsection{The third procedure}
The third procedure is to compute the lower envelope $\calL(\calH_+(p^*))$. As discussed above, for each segment $\overline{ab}$ of $\calL(S)$, 
we need to compute the portion of $H$ between $a$ and $b$, where $H$ is the lower hull of the convex hull whose representing segment contains $\overline{ab}$. To this end, assuming that $a$ is to the left of $b$, it suffices to find the successor vertex of $H$ after $a$ and the predecessor vertex of $H$ before $b$. This can be done using $O(1-\Delta\Phi)$ comparisons, by applying the search lemma on the binary search tree $T_H$ of $H$. After finding the above two vertices, we need to obtain a binary search tree for the portion of $H$ between them; this can be done by splitting $T_H$ at these vertices, which is a procedure that does not cost any comparisons. 

Since $S$ has at most $2K_{p^*}-1$ segments, $\sum_{p^*\in P^*}K_{p^*}=O(n^{4/3})$, and the sum of $-\Delta\Phi$ in the entire algorithm is $O(n\log n)$,  the total number of comparisons needed in the third procedure for all $p^*\in P^*$ is $O(n^{4/3})$. 

\subsubsection{The fourth procedure}
The fourth procedure is to compute $H_+(p^*)$. The idea is to merge pieces of $\calL(\calH_+(p^*))$ one by one from left to right, and each merge is done by first computing the lower common tangent. Note that each piece is a portion of the lower hull of a convex hull of $\calH_+(p^*)$ and is represented by a binary search tree. 
Specifically, suppose we have the lower hull $H_i$ for the first $i$ pieces, i.e., we have a binary search tree $T_i$ for $H_i$. The next step is to merge $H_i$ with the $(i+1)$-th piece $H$ by first computing the lower common tangent of $H_i$ and $H$. Let $T_H$ be the binary search tree for $H$. 
With the lower common tangent, we perform split and merge operations on $T_i$ and $T_H$ for $H$ to obtain a new tree $T_{i+1}$ for the lower hull of the first $i+1$ pieces. Finding the lower common tangent can be done in $O(\log n)$ time~\cite{ref:OvermarsMa81} and constructing the tree $T_{i+1}$ also takes $O(\log n)$ time. 

In the following, we focus on describing a $\Gamma$-procedure that can compute the lower common tangent of $H_i$ and $H$ using $O(1-n^{1/4}\cdot \Delta\Phi)$ comparisons. After the lower common tangent is computed, constructing $T_{i+1}$ does not need any comparisons because split and merge operations on $T_i$ and $T_H$ do not involve any comparisons. 

\paragraph{Notation and algorithm overview.}

For notational convenience, let $H_1=H_i$, $H_2=H$, $T_1=T_i$, and $T_2=T_H$. 

We will use the basic search lemma of the $\Gamma$-algorithm framework of Chan and Zheng~\cite{ref:ChanHo23} discussed above. We will also rely on the technique from Overmars and van Leeuwen~\cite{ref:OvermarsMa81} for computing common tangents of two convex hulls. 

Let the line segment $\overline{t_1t_2}$ be the lower common tangent between $H_1$ and $H_2$, with $t_1\in H_1$ and $t_2\in H_2$. We call $t_1$ and $t_2$ the {\em tangent points}. For ease of exposition, we assume that $\overline{t_1t_2}$ is not collinear with any edge of $H_1\cup H_2$. As such, the tangent points $t_1$ and $t_2$ are unique. 

For each $i=1,2$, since we are searching a vertex $t_i$ from $H_i$, from now on, we view $H_i$ as a sequence of vertices ordered from left to right. We use a {\em chain} to refer to a contiguous subsequence of $H_i$. For two vertices $u$ and $v$ of $H_i$ with $u$ left of $v$, we use $H_i[u,v]$ to denote the chain of $H_i$ from $u$ to $v$ including both $u$ and $v$, while $H_i(u,v)$ is defined to be $H_i[u,v]\setminus\{u,v\}$. Similarly, $H_i[u,v)$ refers to the chain of $H_i$ from $u$ to $v$ including $u$ but excluding $v$; $H_i(u,v]$ is defined likewise. In addition, for any vertex $u$ of $H_i$, we use $u-1$ (resp., $u+1$) to refer to the left (resp., right) neighboring vertex of $u$ in $H_i$; for convenience, if $u$ is the leftmost (resp., rightmost) vertex of $H_i$, then $u-1$ (resp., $u+1$) refers to $u$. 
For example, $H_i(u,v)=H_i[u+1,v-1]$.

The following lemma is based on the technique of Overmars and van Leeuwen~\cite{ref:OvermarsMa81}, which will be used in our algorithm. 

\begin{lemma}\label{lem:prune}{\em (Overmars and van Leeuwen~\cite{ref:OvermarsMa81})}
Given two points $p_1\in H_1$ and $p_2\in H_2$, $p_1$ (resp., $p_2$) partitions $H_1$ (resp., $H_2$) into two chains (we assume that neither chain contains $p_i$, $i=1,2$). Then, among the four chains of $H_1$ and $H_2$, we can determine at least one chain that does not contain a tangent point based on the following six points: $p_1$ (resp., $p_2$) and its two neighboring vertices of $H_1$ (resp., $H_2$); such a chain is called {\em OvL-prunable}\footnote{OvL is the last name initials of Overmars and van Leeuwen.}. 
\end{lemma}

Lemma~\ref{lem:prune} provides a pruning criterion for the binary search algorithm of Overmars and van Leeuwen~\cite{ref:OvermarsMa81} that can compute $\overline{t_1t_2}$ in $O(\log n)$ time. We will use the lemma in a different way. For any subchain of an OvL-prunable chain determined by the lemma, we also say that the subchain is {\em OvL-prunable} (a subchain is a contiguous subsequence of vertices of the chain).

For each $i=1,2$, our algorithm maintains two vertices $u_i$ and $v_i$ of $H_i$ such that $t_i\in H_i(u_i,v_i)$. After $O(1)$ iterations, either $H_i(u_1,v_1)$ or $H_i(u_2,v_2)$ contains a single point. 
The single point is $t_1$ in the former case and $t_2$ in the latter one. 
In either case, one tangent point is determined. After that, the other tangent point can be computed by a similar (and simpler) procedure. For ease of exposition, we assume $t_i$ is not an endpoint of $H_i$, for each $i=1,2$.

\paragraph{Computing a tangent point.}
Initially, for each $i=1,2$, we set $u_i$ and $v_i$ to be the two endpoints of $H_i$, respectively. This guarantees that the algorithm invariant holds, i.e., $t_i\in H_i(u_i,v_i)$. 
We consider a general step of the algorithm. To apply the basic search lemma, we define a predicate $\gamma(u_1,v_1,u_2,v_2)$ with respect to $u_1,v_1,u_2,v_2$, as follows. 

\begin{definition}
Define $\gamma(u_1,v_1,u_2,v_2)$ to be true if neither $H_1(u_1,v_1)$ nor $H_2(u_2,v_2)$ is OvL-prunable when Lemma~\ref{lem:prune} is applied to $(p_1,p_2)$ for any $p_1\in \{u_1,v_1\}$ and $p_2\in \{u_2,v_2\}$.    
\end{definition}

By Lemma~\ref{lem:prune}, whether $\gamma(u_1,v_1,u_2,v_2)$ is true can be determined by at most $O(1)$ points of $L^*$, i.e., at most six points are needed for each pair $(p_1,p_2)$ with $p_1\in \{u_1,v_1\}$ and $p_2\in \{u_2,v_2\}$, 
and thus the predicate is of $O(1)$ degree (where the variables are the coordinates of the $O(1)$ points).


Each iteration of the algorithm proceeds as follows. 
For each $i=1,2$, partition $H_i[u_i,v_i]$ into $r$ chains of roughly equal lengths: $H_i[u_{ij},v_{ij}]$, $1 \leq j\leq r$, for a parameter $r$ to be determined later. Define $C_i=\{H_i[u_{ij}-1,v_{ij}+1]\ |\ 1\leq j\leq r\}$. 
As such, the length of each chain of $C_i$ is roughly equal to the length of $H_i[u_i,v_i]$ divided by $r$. 
Define $\calC=C_1\times C_2$, i.e., $\calC$ consists of $O(r^2)$ pairs of chains such that in each pair, the first chain is from $C_1$ and the second one is from $C_2$. We have the following observation. 

\begin{observation}\label{obser:70}
$\calC$ must contain a pair of chains $H_1[u_{1j_1},v_{1j_1}]$ and $H_2[u_{2j_2},v_{2j_2}]$ such that the predicate $\gamma(u_{1j_1},v_{1j_1},u_{2j_2},v_{2j_2})$ holds true. 
\end{observation}
\begin{proof}
For each $i=1,2$, $t_i$ must be in one of the partitions $H_i[u_{ij},v_{ij}]$, $1 \leq j\leq r$.
By the definition of $C_i$, $t_i$ must be in the interior of a chain $H_i[u_{ij_1},v_{ij_1}]$ of $C_i$. Therefore, $H_i(u_{ij_1},v_{ij_1})$ cannot be an OvL-prunable chain. Hence, 
$\gamma(u_{1j_1},v_{1j_1},u_{2j_2},v_{2j_2})$ must hold true. 
\end{proof}

In light of the above observation, we apply the basic search lemma on the $O(r^2)$ cells of $\calC$ to find a pair of chains $H_1[u_{1j_1},v_{1j_1}]$ and $H_2[u_{2j_2},v_{2j_2}]$ such that the predicate $\gamma(u_{1j_1},v_{1j_1},u_{2j_2},v_{2j_2})$ is true.  
By the basic search lemma, this can be accomplished using $O(1-r^2\cdot \Delta\Phi)$ comparisons. The subsequent lemma, which proves a key property, ensures our ability to recursively employ the algorithm.


\begin{lemma}\label{lem:80}
If $\gamma(u_{1j_1},v_{1j_1},u_{2j_2},v_{2j_2})$ is true, then either $t_1\in H_1[u_{1j_1},v_{1j_1}]$ or $t_2\in H_2[u_{2j_2},v_{2j_2}]$, and which case happens can be determined using $O(1)$ comparisons. 
\end{lemma}
\begin{proof}
For each $i=1,2$, let $a_i$ and $b_i$ be the leftmost and rightmost vertices of $H_i$, respectively. The two vertices $u_{ij_i}$ and $v_{ij_i}$ partition $H_i$ into three chains $H_i[a_i,u_{ij_i})$, $H_i[u_{ij_i},v_{ij_i}]$, and $H_i(v_{ij_i},b_i]$. Due to that $\gamma(u_{1j_1},v_{1j_1},u_{2j_2},v_{2j_2})$ is true, we argue below that for either $i=1$ or $i=2$, both $H_i[a_i,u_{ij_i})$ and $H_i(v_{ij_i},b_i]$ are OvL-prunable, and therefore, $t_i$ must be in $H_i[u_{ij_i},v_{ij_i}]$. This is done by case analysis after applying Lemma~\ref{lem:prune} on all pairs of points $(p_1,p_2)$ with $p_1\in \{u_{1j_1},v_{1j_1}\}$ and $p_2\in \{u_{2j_2},v_{2j_2}\}$.

Indeed, suppose we first apply Lemma~\ref{lem:prune} on $(u_{1j_1},u_{2j_2})$. Since $\gamma(u_{1j_1},v_{1j_1},u_{2j_2},v_{2j_2})$ is true, neither $H_1(u_{1j_1},v_{1j_1})$ nor $H_2(u_{2j_2},v_{2j_2})$ is OvL-prunable. Since $H_1(u_{1j_1},v_{1j_1})\subseteq H_1(u_{1j_1},b_1]$ and $H_2(u_{2j_2},v_{2j_2})\subseteq H_2(u_{2j_2},b_2]$, we obtain that neither $H_1(u_{1j_1},b_1]$ nor $H_2(u_{2j_2},b_2]$ is OvL-prunable. Therefore, at least one of $H_1[a_1,u_{1j_1})$ and $H_2[a_2,u_{2j_2})$ is OvL-prunable. Without loss of generality, we assume that $H_1[a_1,u_{1j_1})$ is OvL-prunable (see Fig.~\ref{fig:prunable}). 

\begin{figure}[h]
\begin{minipage}[t]{\textwidth}
\begin{center}
\includegraphics[height=1.2in]{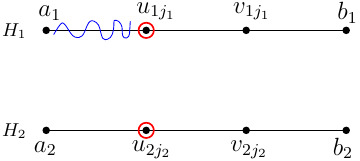}
\caption{\footnotesize After Lemma~\ref{lem:prune} is applied on $(u_{1j_1},u_{2j_2})$, $H_1[a_1,u_{1j_1})$ becomes OvL-prunable.}
\label{fig:prunable}
\end{center}
\end{minipage}
\end{figure}

Next we apply Lemma~\ref{lem:prune} on $(v_{1j_1},u_{2j_2})$. Following similar analysis to the above, one of $H_1(v_{1j_1},b_1]$ and $H_2[a_2,u_{2j_2})$ must be OvL-prunable. If $H_1(v_{1j_1},b_1]$ is OvL-prunable, then we have proved that both $H_i[a_i,u_{ij_i})$ and $H_i(v_{ij_i},b_i]$ are OvL-prunable for $i=1$ and thus we are done with proof. In the following, we assume that $H_1(v_{1j_1},b_1]$ is not OvL-prunable and thus $H_2[a_2,u_{2j_2})$ must be OvL-prunable (see Fig.~\ref{fig:prunable10}).

\begin{figure}[h]
\begin{minipage}[t]{\textwidth}
\begin{center}
\includegraphics[height=1.2in]{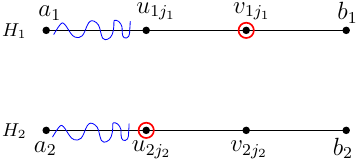}
\caption{\footnotesize After Lemma~\ref{lem:prune} is applied on $(v_{1j_1},u_{2j_2})$, $H_2[a_2,u_{2j_2})$ becomes OvL-prunable.}
\label{fig:prunable10}
\end{center}
\end{minipage}
\end{figure}

Now we apply Lemma~\ref{lem:prune} on $(v_{1j_1},v_{2j_2})$. Following similar analysis to the above, one of $H_1(v_{1j_1},b_1]$ and $H_2(v_{2j_2},b_2]$ must be OvL-prunable. Without loss of generality, we assume that $H_1(v_{1j_1},b_1]$ is OvL-prunable. As such, we have proved that both $H_i[a_i,u_{ij_i})$ and $H_i(v_{ij_i},b_i]$ are OvL-prunable for $i=1$ and thus we are done with proof.

In summary, the above proves that both $H_i[a_i,u_{ij_i})$ and $H_i(v_{ij_i},b_i]$ are OvL-prunable for some $i\in \{1,2\}$, and therefore, $t_i$ must be in $H_i[u_{ij_i},v_{ij_i}]$. The above also provides a procedure that can determine the $i$ such that $t_i\in H_i[u_{ij_i},v_{ij_i}]$ by applying Lemma~\ref{lem:prune} on $O(1)$ pairs of points, which uses $O(1)$ comparisons. The lemma thus follows. 
\end{proof}

By virtue of the above lemma, without loss of generality, we assume that $t_1\in H_1[u_{1j_1},v_{1j_1}]$ (and thus $t_1\in H_1(u_{1j_1}-1,v_{1j_1}+1)$). This finishes the current iteration. Next, we proceed on the next iteration with the two chains $H_1[u_{1j_1}-1,v_{1j_1}+1]$ and $H_2[u_2,v_2]$, i.e., we update $u_1=u_{1j_1}-1$ and $v_1=v_{1j_1}+1$ (and thus $t_1\in H_1(u_1,v_1)$). Recall that $t_2\in H_2(u_2,v_2)$. Hence, the algorithm invariant holds. 

In this way, each iteration on $H_1[u_1,v_1]$ and $H_2[u_2,v_2]$ shrinks one of them to a subchain of size roughly $1/r$ of its original size. Since $|H_1|+|H_2|\leq n$, after $O(\log_rn)$ iterations, 
a tangent point can be found. The total number of comparisons is thus $O(\log_rn-r^2\cdot \Delta\Phi)$. 

\paragraph{Computing the second tangent point}
Without loss of generality, assume that $t_1$ has been found as above. Then, to locate $t_2$, we can basically follow a similar but simpler process, with the help of $t_1$. We briefly discuss it below. 

We start with the following simpler version of Lemma~\ref{lem:prune}.

\begin{lemma}\label{lem:simprune}{\em (Overmars and van Leeuwen~\cite{ref:OvermarsMa81})}
Given a point $p_2\in H_2$, $p_2$ partitions $H_2$ into two chains (we assume that neither chain contains $p_2$). We can determine at least one chain that does not contain the tangent point $t_2$ based on the following four points: $t_1$, $p_2$ and its two neighboring vertices of $H_2$; such a chain is called {\em OvL-prunable}. 
\end{lemma}

For a chain $H_2[u,v]$ of $H_2$, we define a predicate $\gamma(u,v)$ with respect to the two vertices $u$ and $v$, as follows. 

\begin{definition}
Define $\gamma(u,v)$ to be true if $H_2(u,v)$ is not OvL-prunable when Lemma~\ref{lem:simprune} is applied to $p_2$ for any $p_2\in \{u,v\}$. 
\end{definition}

As before, the predicate  $\gamma(u,v)$ is of $O(1)$ degree.

Let $H_2[u_2,v_2]$ be the portion of $H_2$ after $t_1$ is computed in the above algorithm. By the algorithm invariant, we have $t_2\in H_2(u_2,v_2)$. 

We are now in a position to describe the algorithm for finding $t_2$, which works iteratively. 
In each iteration, we start with a chain $H_2[u,v]$ such that $t_2\in H_2(u,v)$. Initially, set $u=u_2$ and $v=v_2$. Hence, the algorithm invariant holds.

In each iteration, we do the following. 
We partition $H_2[u,v]$ into $r$ chains of roughly equal lengths: $H_2[u_j,v_j]$, $1\leq j\leq r$, with the same $r$ as above. 
Define $C=\{H_2[u_{j}-1,v_{j}+1]\ |\ 1\leq j\leq r\}$. 
We have the following observation. The argument is similar to Observation~\ref{obser:70} and thus is omitted. 
\begin{observation}
$C$ must contain a chain $H_2[u_{i},v_{i}]$ such that the predicate $\gamma(u_i,v_i)$ holds true. 
\end{observation}

By the preceding observation, we apply the basic search lemma on the $r$ chains of $C$ to find a interval $H_2[u_i,v_i]$ so that $\gamma(u_i,v_i)$ holds true. We have the following lemma. 

\begin{lemma}\label{lem:110}
If $\gamma(u_i,v_i)$ holds true, then $t_2$ must be in $H_2[u_i,v_i]$.
\end{lemma}
\begin{proof}
Let $a_2$ and $b_2$ be the left and right endpoints of $H_2$, respectively. The two vertices $u_i$ and $v_i$ partition $H_2$ into three chains: $H_i[a_2,u_i)$, $H_i[u_i,v_i]$, and $H(v_i,b_2]$. In the following, we argue that both $H[a_2,u_i)$ and $H(v_i,b_2]$ are OvL-prunable, and therefore, $t_2$ must be in $H_2[u_i,v_i]$.

Suppose we apply Lemma~\ref{lem:simprune} on $u_i$. By Lemma~\ref{lem:simprune}, either $H_2[a_2,u_i)$ or $H_2(u_i,b_2]$ must be OvL-prunable. Since $\gamma(u_i,v_i)$ is true, by definition, $H_2(u_i,v_i)$ is not OvL-prunable. As $H_2(u_i,v_i)\subseteq H_2(u_i,b_2]$, $H_2(u_i, b_2)$ cannot be OvL-prunable. Therefore, $H_2[a_2,u_i)$ must be OvL-prunable. 

Suppose we apply Lemma~\ref{lem:simprune} on $v_i$. Following a similar analysis to the above, $H(v_i,b_2]$ must be OvL-prunable. 
\end{proof}

We proceed with the next iteration on the chain $H_2[u_i-1,v_i+1]$, i.e., update $u=u_i-1$ and $v=v_i+1$. By Lemma~\ref{lem:110}, $t_2$ must be in $H_2(u,v)$ and thus the algorithm invariant is established. 

Since each iteration shrinks the chain to roughly $1/r$ of its original size, after $O(\log_r n)$ iterations, $t_2$ will be located. 
As such, the number of comparisons is $O(\log_rn-r\cdot \Delta\Phi)$.

\paragraph{Summary}
In summary, with $O(\log_rn - r^2\cdot \Delta\Phi)$ comparisons, the lower common tangent $\overline{t_1t_2}$ can be computed. Setting $r=n^{1/8}$ leads to an upper bound of $O(1-n^{1/4}\cdot \Delta\Phi)$ on the number of comparisons.

In this way, the total number of comparisons for constructing the lower hull $H_+(p^*)$ by merging the pieces of $\calL(\calH_+(p^*))$ is $O(K_{p^*}-n^{1/4}\cdot \Delta\Phi)$ since $\calL(\calH_+(p^*))$ has at most $2K_{p^*}-1$ pieces.
As $\sum_{p^*\in P^*}K_{p^*}=O(n^{4/3})$ and the sum of $-\Delta\Phi$ in the entire algorithm is $O(n\log n)$, we obtain that the total number of comparisons for computing the lower hulls $H_+(p^*)$ for all $p^*\in P^*$ is bounded by $O(n^{4/3})$.

\subsection{Solving the subproblems $\boldsymbol{T(n^{2/3},n^{1/3})}$}
In this section, we show that each subproblem $T(n^{2/3},n^{1/3})$ in Recurrence~\eqref{equ:50} can be solved using $O(n^{2/3})$ ``amortized'' comparisons, or all $O(n^{2/3})$ subproblems $T(n^{2/3},n^{1/3})$ in~\eqref{equ:50} can be solved using a total of $O(n^{4/3})$ comparisons.

Recall that $P$ is the set of $n$ points and $L$ is the set of $n$ lines for the original problem in Recurrence~\eqref{equ:50}. For notational convenience, let $m=n^{1/3}$ and thus we want to
solve $T(m^2,m)$ using $O(m^2)$ amortized comparisons. More specifically, we are given a set $P'$ of $m^2$ points and a set $L'$ of 
$m$ lines. The problem is to compute all non-empty faces of the arrangement $\calA(L')$ of the lines of $L'$ that contain at least one point of $P'$. Our algorithm will return a binary search tree representing each non-empty face. 

We first construct the arrangement $\calA(L')$, which can be done in $O(m^2)$ time~\cite{ref:EdelsbrunnerAr92}. Then, for each face of $\calA(L')$, we construct a binary search tree representing the face. Doing this for all faces takes $O(m^2)$ time as well. As such, it suffices to locate the face of $\calA(L')$ containing each point of $P'$. Note that this point location problem is exactly a bottleneck subproblem in the algorithm for Hopcroft's problem in \cite{ref:ChanHo23}. Applying the $\Gamma$-algorithm of \cite[Section~4.3]{ref:ChanHo23}, we can solve all $O(n^{2/3})$ subproblems $T(n^{2/3},n^{1/3})$ in Recurrence~\eqref{equ:50} using a total of $O(n^{4/3})$ comparisons. For completeness, we present our own algorithm in the following. 

We construct a hierarchical $(1/m)$-cutting for the lines of $L'$: $\Xi_0,\Xi_1,\ldots,\Xi_k$, with a constant $\rho$ as discussed in Section~\ref{sec:pre}. This takes $O(m^2)$ time~\cite{ref:ChazelleCu93}. Since the last cutting $\Xi_k$ is an $(1/m)$-cutting, no line of $L'$ crosses the interior of $\sigma$, for each cell $\sigma$ of $\Xi_k$. As such,  every cell $\sigma$ of $\Xi_k$ is completely contained in a face of $\calA(L')$. By the virtue of this property, for each point $p\in P'$, once we locate the cell of $\Xi_k$ containing $p$, we immediately know the face of $\calA(L')$ containing $p$, provided that we label each cell $\sigma$ of $\Xi_k$ by the name of the face of $\calA(L')$ containing $\sigma$; the latter task can be done in $O(m^2)$ time~\cite{ref:ChazelleCu93} (e.g., by slightly modifying the cutting construction algorithm in \cite{ref:ChazelleCu93}). 

In this way, we essentially reduce our problem to point locations in the cutting $\Xi_k$: For each point $p\in P'$, locate the cell of $\Xi_k$ containing $p$. Using the hierarchical cutting, traditional approach can solve the problem in $O(m^2\log m)$ time since $|P'|=O(m^2)$ and performing point location for each point of $P'$ takes $O(\log m)$ time. This would result in an overall $O(n^{4/3}\log n)$ comparisons for solving all subproblems in Recurrence~\eqref{equ:50}. In what follows, we show that using the basic search lemma, there is a $\Gamma$-algorithm that can solve the problem using $O(m^2)$ amortized comparisons. This is done by using a (not contiguous) subsequence of only $O(1)$ cuttings in the hierarchy.  

For any $0\leq i\leq k$, let $|\Xi_i|$ represent the number of cells of $\Xi_i$. Recall from Section~\ref{sec:pre} that $|\Xi_i|=O(\rho^{2i})$. 
Define $i_1,i_2,\ldots,i_g$ as a subsequence of the indices $0,1,\ldots,k$ such that (1) $i_g=k$; (2) for any $1\leq j\leq g-1$, suppose $i_{j+1}$ is already defined; then $i_j$ is defined in such a way that $\rho^{2i_{j+1}}/\rho^{2i_j}=\Theta(m^{1/2})$. Note that it is easy to find such a subsequence in $O(k)$ time by scanning the indices $0,1,\ldots,k$ barckwards, starting from having $i_g=k$. In addition, let $i_0=0$. 
Observe that for each $j$ with $0\leq j\leq g-1$, each cell of $\Xi_{i_j}$ contains $O(m^{1/2})$ cells of $\Xi_{i_{j+1}}$. Also, since $\rho^k=O(m^2)$, we have $g=O(1)$. Our problem is to find the cell of $\Xi_{i_g}$ (which is $\Xi_k$) that contains $p$ for all points $p\in P'$. 

For each point $p\in P'$, starting from $\Xi_{i_0}$, for all $j=0,1,\ldots, g$ in this order, we find the cell of $\Xi_{i_j}$ containing $p$, as follows. Suppose we know the cell $\sigma$ of $\Xi_{i_j}$ containing $p$. To find the cell of $\Xi_{i_{j+1}}$ containing $p$, since $\sigma$ contains $O(m^{1/2})$ cells of $\Xi_{i_{j+1}}$, we apply the basic search lemma on these cells, which costs $O(1-m^{1/2}\cdot \Delta\Phi)$ comparisons. Hence, it eventually takes $O(g-m^{1/2}\cdot \Delta\Phi)$ comparisons to find the cell of $\Xi_{i_g}$ containing $p$. As $g=O(1)$ and $|P'|=m^2$, the total number of  comparisons for all points of $P'$ is $O(m^2-m^{1/2}\cdot \Delta\Phi)$.

Now back to our original problem for solving the subproblem $T(n^{2/3},n^{1/3})$ in Recurrence~\eqref{equ:50}, the above shows that each subproblem can be solved using $O(n^{2/3}-n^{1/6}\cdot \Delta\Phi)$ comparisons as $m=n^{1/3}$. Therefore, solving all $O(n^{2/3})$ subproblems in Recurrence~\eqref{equ:50} can be done using $O(n^{4/3}-n^{1/6}\cdot \Delta\Phi)$ comparisons. Since  the sum of $-\Delta\Phi$ in the entire algorithm is $O(n\log n)$, we obtain that the total number of comparisons for solving all $O(n^{2/3})$ subproblems $T(n^{2/3},n^{1/3})$ in Recurrence~\eqref{equ:50} is bounded by $O(n^{4/3})$.

We remark that using a much simpler fractional cascading technique for line arrangements~\cite[Section~3]{ref:ChanHo23}, we can solve all $O(n^{2/3})$ subproblems $T(n^{2/3},n^{1/3})$ in $O(n^{4/3})$ randomized expected time without resorting to $\Gamma$-algorithms. 

\subsection{Putting it all together}

The above proves Lemma~\ref{lem:40}, and thus $T(n,n)$ in Recurrence \eqref{equ:50} can be bounded by
$O(n^{4/3})$ after $O(2^{\poly(n)})$ time preprocessing, as discussed before. Equivalently,
$T(b,b)$ in Recurrence~\eqref{equ:40} is $O(b^{4/3})$ after $O(2^{\poly(b)})$ time
preprocessing. Notice that the preprocessing work is done only once and for all
subproblems $T(b,b)$ in \eqref{equ:40}.
Since $b=(\log\log n)^3$, we have $2^{\poly(b)}=O(n)$.
As such, $T(n,n)$ in \eqref{equ:40} solves to $O(n^{4/3})$ and we thus obtain the
following theorem.

\begin{theorem}\label{theo:symmetric}
Given a set of $n$ points and a set of $n$ lines in the plane, we can report all faces of the line arrangement 
that contain at least one point in $O(n^{4/3})$ time.
\end{theorem}

\subsection{The asymmetric case and the lower bound}

The asymmetric case of the problem is handled in the following corollary, which uses the algorithm of Theorem~\ref{theo:symmetric} as a subroutine. 

\begin{corollary}\label{coro:asymetric}
Given a set of $m$ points and a set of $n$ lines in the plane, we can report all faces of the line arrangement 
that contain at least one point in $O(m^{2/3}n^{2/3}+(n+m)\log n)$ time.
\end{corollary}
\begin{proof}
Let $P$ be the set of $m$ points and $L$ the set of $n$ lines. 
Depending on whether $m\geq n$, there are two cases. 

\begin{enumerate}
  \item If $m\geq n$, depending on whether $m< n^2$, there are two subcases.

  \begin{enumerate}
     \item If $m\geq n^2$, then we use the straightforward algorithm discussed in Section~\ref{sec:intro}, i.e., first construct the arrangement $\calA(L)$ and then perform point locations for points of $P$. The runtime is $O(m\log n+n^2)$, which is $O(m\log n)$ since $m\geq n^2$. 
    
    \item If $m<n^2$, then set $r=m/n$ so that $n/r = m/r^2$. Consequently, applying \eqref{equ:10} and solving each subproblem $T(m/r^2,n/r)$ by Theorem~\ref{theo:symmetric} give us $T(m,n)=O(m^{2/3}n^{2/3}+m\log n)$.

  \end{enumerate}

  \item If $n>m$, depending on whether $n\geq m^2$, there are two subcases.
  \begin{enumerate}
      \item If $n\geq m^2$, then setting $r=m$ and applying \eqref{equ:20} give us $T(m,n)=O(n\log n)+O(m^2)\cdot T(1,n/m^2)$. Note that $T(1,n/m^2)$ is to compute the cell of the line arrangement of $O(n/m^2)$ lines that contains a single point, which can be solved in $O(n/m^2\cdot \log(n/m^2))$ time. As such, we obtain $T(m,n)=O(n\log n)$. 

      \item If $n<m^2$, then set $r=n/m$ so that $m/r=n/r^2$. Consequently, applying \eqref{equ:20} and solving each subproblem $T(m/r,n/r^2)$ by Theorem~\ref{theo:symmetric} give us $T(m,n)=O(m^{2/3}n^{2/3}+n\log n)$.
  \end{enumerate}

\end{enumerate}

Combining all above cases leads to an upper bound of $O(m^{2/3}n^{2/3}+(n+m)\log n)$ on the total time of the overall algorithm.
\end{proof}


We finally prove the lower bound below, which justifies the optimality of Corollary~\ref{coro:asymetric}.

\begin{theorem}
Given a set of $m$ points and a set of $n$ lines in the plane, computing all non-empty faces of the line arrangement 
requires $\Omega(m^{2/3}n^{2/3}+(n+m)\log n)$ time under the algebraic decision tree model.
\end{theorem}
\begin{proof}
First of all, as discussed in Section~\ref{sec:intro}, $\Omega(m^{2/3}n^{2/3}+n)$ is a lower bound because it is a lower bound on the output size in the worst case~\cite{ref:EdelsbrunnerOn86}. $\Omega(n\log n)$ is also a lower bound because computing a single face requiring that much time under the algebraic decision tree model. It remains to show that $\Omega(m\log n)$ is a lower bound.

Consider the case where $m=\Omega(n^2)$. Let $P$ be the set of $m$ points and $L$ the set of $n$ lines. Let $k$ be the number of cells in the arrangement $\calA(L)$. In general, $k=\Theta(n^2)$. To simplify the notation, let $m=k$. Our problem is to compute all non-empty faces of $\calA(L)$. 

We arbitrary assign indices $1,2,\ldots,m$ to cells of $\calA(L)$. For any point $p$ in the plane, if $p$ is on the boundary of multiple faces (e.g., $p$ is on an edge or a vertex of $\calA(L)$), then we assume $p$ is only contained in the face that has the largest index. In this way, every point in the plane belongs to a single face of $\calA(L)$. Define $I(p)$ as the index of the face of $\calA(L)$ containing $p$. 

Let $P=\{p_1,p_2,\ldots,p_m\}$.
Consider the following problem, referred to as the {\em distinct-face problem}: Decide whether points of $P$ are in distinct faces of $\calA(L)$, i.e., whether $I(p_i)\neq I(p_j)$ for any two points $p_i,p_j\in P$. Clearly, the problem can be reduced to our original problem in $O(m)$ time. Indeed, if we have a solution to our original problem, then the distinct-face problem can be easily solved in $O(m)$ time. In what follows, using Ben-Or's techniques~\cite{ref:Ben-OrLo83}, we argue that solving the distinct-face problem requires  $\Omega(m\log n)$ time under the algebraic decision tree model. This implies the same lower bound for our original problem. 
 
By slightly abusing the notation, we also use $p_i$ to refer to its two coordinates. Define the set
$$W=\{(p_1,p_2,\ldots,p_m)\in \bbR^{2m}\ |\ \prod_{i\neq j}(I(p_i)-I(p_j)) \neq 0\}.$$
Our distinct-face problem is equivalent to the following {\em membership problem}: For any point $P=(p_1,p_2,\ldots,p_m)\in \bbR^{2m}$, decide whether $P\in W$. In the following, we argue that $W$ has $\Omega(m!)$ path-connected components. According to Ben-Or's result~\cite{ref:Ben-OrLo83}, this implies that $\Omega(\log (m!))=\Omega(m\log m)$ is a lower bound for solving the distinct-face problem under the algebraic decision tree model. 

For any permutation $\sigma$ of $1,2,\ldots,m$, define
$$W_{\sigma}=\{(p_1,p_2,\ldots,p_m)\in \bbR^{2m}\ |\  I(p_{\sigma(1)})<I(p_{\sigma(2)})<\ldots<I(p_{\sigma(m)})\}.$$
Clearly, $W_{\sigma}\subseteq W$. We claim that each $W_{\sigma}$ belongs to a distinct connected component of $W$. Indeed, consider the collection of functions $f_{ij}(p_1,p_2,\ldots,p_m)=\text{sign}(I(p_i)-I(p_j))$ for $i\neq j$, i.e., $f_{ij}=1,0,-1$ if $I(p_i)-I(p_j)$ is positive, zero, and negative, respectively. For any $\sigma$, each of these functions is constant either equal to $1$ or $-1$ over $W_{\sigma}$. Now consider two permutations $\sigma_1$ and $\sigma_2$, and two points $P_1,P_2\in \bbR^{2m}$ with $\mathbb{P}_1\in W_{\sigma_1}$ and $\mathbb{P}_2\in W_{\sigma_2}$. Due to the continuity, any path in $\bbR^2$ connecting $\mathbb{P}_1$ and $\mathbb{P}_2$ must go through a point $\mathbb{P}$ in which the value of at least one of the functions $f_{ij}(\mathbb{P})$ is $0$, and thus the point $\mathbb{P}$ does not belong to $W$. Therefore, $W_{\sigma_1}$ and $W_{\sigma_2}$ are in two different connected components of $W$. 

In light of the claim, the number of connected components of $W$ is no less than the number of permutations of $1,2,\ldots,m$, which is $m!$. 

The above proves that $\Omega(m\log n)$ is a lower bound for our original problem. The theorem thus follows. 
\end{proof}

\bibliographystyle{plainurl}
\bibliography{reference}




\end{document}